\documentclass[11pt]{article}
 \usepackage{epsfig,amssymb,graphicx,amsmath,amsthm,subcaption}
 \usepackage{tikz}
 \usetikzlibrary{shapes.geometric, arrows}
 \usetikzlibrary{positioning}
 \usepackage{epstopdf}
 \usepackage{enumerate}
 \usepackage{bm}

 \newtheorem{theorem}{Theorem}[section]

 \newtheorem*{remark}{Remark}

 \tikzstyle{H} = [circle, minimum width=1.5cm, minimum height=1cm,text centered, draw=black, fill=magenta!30]
 \tikzstyle{S} = [circle, minimum width=1.5cm, minimum height=1cm,text centered, draw=black, fill=blue!30]
 \tikzstyle{I} = [circle, minimum width=1.5cm, minimum height=1cm,text centered, draw=black, fill=red!30]
 \tikzstyle{W} = [circle, minimum width=1.5cm, minimum height=1cm,text centered, draw=black, fill=gray!30]
 \tikzstyle{R} = [circle, minimum width=1.5cm, minimum height=1cm,text centered, draw=black, fill=yellow!30]
 \tikzstyle{line} = [thick,-,>=stealth]
 \tikzstyle{arrow} = [thick,->,>=triangle 90]
 \tikzstyle{openA} = [thick,<-,>=open triangle 90]
 \tikzstyle{openB} = [thick,->,>=open triangle 90]
 \pgfarrowsdeclarecombine{twotriang}{twotriang}%
 {triangle 90}{triangle 90}{triangle 90}{triangle 90}
 \tikzstyle{Darrow} = [thick,double,->,>=triangle 90]
 \tikzstyle{sarrow} = [thick,->,>=triangle 90,shorten >=2cm]
 \tikzstyle{curveA} = [pil/.style={->,thick,shorten <=2pt,shorten >=2pt,}]
 
\begin{document}
 	
 		\title{Mathematical model of plant-virus interactions mediated by RNA interference}
 		
 \author{G. Neofytou,\hspace{0.5cm}Y.N. Kyrychko,\hspace{0.5cm}K.B. Blyuss\thanks{Corresponding author. Email: k.blyuss@sussex.ac.uk} 
\\\\ Department of Mathematics, University of Sussex, Falmer,\\
Brighton, BN1 9QH, United Kingdom}

\maketitle
 		
 		\begin{abstract}
 			Cross-protection, which refers to a process whereby artificially inoculating a plant with a mild strain provides protection against a more aggressive isolate of the virus, is known to be an effective tool of disease control in plants. In this paper we derive and analyse a new
 			mathematical model of the interactions between two competing viruses with particular account for RNA interference. Our results show that  co-infection of the host can either increase or decrease the potency of individual infections depending on the levels of cross-protection or cross-enhancement between different viruses. Analytical and numerical bifurcation analyses are employed to
 			investigate the stability of all steady states of the model in order to identify parameter
 			regions  where the system exhibits synergistic or antagonistic behaviour between viral strains, as well as different types of host recovery. We show that not only viral attributes but also the propagating component of RNA-interference in plants can play an important role in determining the dynamics.
 		\end{abstract}
 		
 	\section{Introduction}
 	
 	With a projected number of 9.7 billion people by the year 2050, the world population and its continuing growth is heavily dependent on a steady agricultural output in order to provide a sustainable food source.  In light of the agricultural stagnation experienced in the last decade, further fuelled by the public opposition to controversial newer practices \cite{Kuntz2014}, securing an adequate and reliable food source has never been more relevant. It is estimated that up to 40$\%$ of global crop production is lost due to pathogens, animals and weeds \cite{Savary2012}. This inevitably led to the development of different agricultural practices including the use of various pesticides and ultimately genetic engineering. Although significant efforts are made to increase crop yield and with a good degree of success, perhaps, a more effective or environmentally safe way to address this problem  hides in better investigating and  understanding methods that are currently employed. In this respect, mathematical models can provide invaluable insights into the dynamics of plant infections and allow better control over agricultural losses.
 	
 	Similar to the studies of infectious disease in humans,  mathematical modelling allows one to investigate how an infection propagates within a population of plants. As such, the interactions between healthy and infected plants can usually be described by empirically derived relationships between plants and an insect population which acts as the disease vector and is comparable to epidemic models of mosquito-borne diseases in humans \cite{Purcell2005}. Several mathematical models have also analysed the efficiency of simpler and more traditional methods of fighting plant infection, such as roguing and replanting, in which any plants afflicted by the disease are simply removed and replaced by other healthy plants \cite{Chan1994,VandenBosch1996,Zhang2012a}.
 	
    In the 1970s, the increase of computing power allowed the development of models capable of simulating vector population and weather conditions \cite{Gutierrez1974,Frazer1977,Kirtani1978,Irwin2000}. Despite their simple structure, these models enabled the integration of various disease control options, thus  creating a framework where such methods could be  analysed and evaluated. Madden et al. \cite{Madden2000} have performed a detailed analysis of the transitional dynamics of plant diseases taking into account the effects of vector emigration.  Depending on the way they are transmitted, plant viruses are classified as non-persistent, semi-persistent and persistent, and Madden et al. \cite{Madden2000} demonstrated which of these three classes were more susceptible to changes in vector longevity and inoculation, acquisition rates and vector mobility. Subsequent models have looked into the transmission dynamics of a pair of "helper" and helper-dependent viruses.  Zhang et al. \cite{Zhang2000} provided insights into the commonly observed phenomenon where infecting a host with only a helper virus would cause minimal or no damage to the host, whereas, additionally introducing the helper-dependent virus would produce far more devastating symptoms.

	In the last few decades it has been discovered that viruses employ a wide antigenic diversity as an effective strategy to survive within the host population \cite{Frank2002,Lipsitch2007}. By employing a variety of antigenic ally distinct strains, viruses are able to adapt sufficiently fast to evade the host's immune system. Antigenic variation is known to be effective for a large number of pathogens affecting humans, including malaria \cite{Gupta1994,Ferreira2004}, meningitis \cite{Gupta1996,Gupta1999}, dengue fever \cite{Gog2002} and influeza \cite{Ferguson2003}. The interactions between multiple strains are generally classified as either an ecological interference, or an immunological interference. The first type of interactions describes a simple case where individual hosts can only be infected with a single strain, and subsequently the are removed from the population susceptible to other strains \cite{Levin2004}.  Immunological interference corresponds to situations where infection with one strain may cause partial or full immunity to the remaining strains \cite{Gupta1999}, or sometimes it can even augment the susceptibility of the host and the transmissibility of other strains \cite{Recker2009}. To better understand the dynamics of multi-strain diseases, a large number of mathematical models have been developed that can be divided into individual-based and equation-based models. In individual-based models, all pathogen strains are treated as individuals interacting according to a fixed set of rules \cite{Ferguson2003,Buckee2004,Buckee2010,Cisternas2004}, whereas in equation-based models, hosts are categorised either according to preceding exposure to individual strains \cite{Andreasen1997,Gomes2002}, or based on their immunity to specific strains \cite{Gog2002,Kryazhimskiy2007}.
	
 	One very efficient way of protecting a plant against a disease known as {\it cross-protection}, consists of the process by which prior infection of the plant with a primary virus can prevent or interfere with the subsequent infection with a secondary virus of the same family \cite{Zhou2012}. In such a case, deliberately infecting the plant with a less virulent strain can offer protection against a much more virulent isolate of the virus. Although this natural phenomenon was first demonstrated more than 80 years ago, its precise mechanisms are still not fully understood, and several hypotheses have been put forward to explain how cross-protection works \cite{Pennazio2001}. It has been suggested that the primary infection could trigger the formation of specific antibodies which could prevent the subsequent infection by a similar virus. Another possibility is the coat-protein mediated resistance that is usually expressed by transgenic plants encoding viral coat-proteins. However, in the case of competing viral strains, the coat protein of the primary strain can also interfere with the encapsidation process of the secondary strain, thus rendering it ineffective for cell-to-cell transmission \cite{Beachy1999,Bendahmane1997}. Additionally, if the two viruses are closely related they could very well be competing for the same components which are essential for viral replication, or that the occupation of replication sites by the primary strain could cause a spatial exclusion of the secondary strain \cite{Lee2005,Takeshita2004,Gal-On2006}. 
 	
 	A very promising explanation of cross-protection can be found in the biological pathway known as a {\it post-transcriptional gene silencing}, or {\it RNA-interference} (RNAi) \cite{Ratcliff1999}. This mechanism is characterized by the ability of cells to recognise and degrade the messenger RNA of invading RNA viruses  or cause the  methylation of target gene sequences and the genome of DNA viruses \cite{Waterhouse1999,Escobar2000,Sijen2000}. This process is mediated by different lengths of double stranded RNAs (dsRNA) that are generated by an inverted-repeat transgene or an invading virus during its replication process. A very simple description of the core pathway is as follows. The presence of transgenic or viral dsRNA triggers an immune response within the host cell, whereby the foreign RNA is targeted by specialized enzymes called dicers (DLC) which cleave it into short 21-26 nucleotide long molecules. These molecules, named short interfering RNAs (siRNA) or microRNA(miRNA) can then be used to assemble a special protein complex called RNA-induced silencing complex (RISC) which has the capacity to recognise and degrade RNAs containing complementary sequences. By doing so, viral replication is prohibited, and, therefore, it prevents the spread of infection \cite{AlessandraTenorioCosta2013,Hammond2000,Bernstein2001}. It is very important to note that the siRNA can also be transported into neighbouring cells, thus acting as a mobile warning signal that can fortify and prepare cells by allowing them to express the antiviral components even before they become infected \cite{Zhang2012a,Wassenegger2000,Zhang2012b}.  
 	
 	The ability to induce a propagating warning signal can most likely be attributed to the evolutionary race between the plant and the viruses that afflict them, as it has been demonstrated that viruses can suppress different stages of the RNA-interference pathway \cite{AlessandraTenorioCosta2013,Pumplin2013,Raja2008}. In some cases the virus can prevent degradation of its genome by either suppressing cellular innate immune response  or by simply managing to successfully spread before  being detected. The latter can be achieved by  moving  into another cell before a specific threshold of viral dsRNA has accumulated, and one that is necessary in order for the cell to initiate a response. In other cases, the virus can only suppress the propagating warning signal, therefore, depending on which component of the immune response is targeted by viral suppressors, one can expect a different phenotype of recovery. 
 	
 	It is important to note that in the studies of plant pathology, single-host interactions between different viruses  are highly important as they can often produce distinct types of  host immune response. Therefore while some viral pairs are able to  facilitate each other and engage in a synergistic  relationship others will compete with each other for dominance \cite{Malapi-Nelson2009,Wege2007,Pruss1997}. Contrarily to cross-protection, enhanced symptom display  occurs  when plants co-infected with two or multiple viral strains experience symptoms that are more severe to the single-strain example and often exhibit an elevated viral load for one or multiple viruses. Therefore, depending on the level of competition between the viruses and the corresponding immune response a different degree of cross-protection or cross-enhancement can be observed.
 	
 	It is most unlikely that any synergistic or antagonistic outcome of a viral co-infection in a single host, associated with cross-protection or enhanced symptom display, can be fully explained by one single mechanism. This is due to the wide variety of plants with an immune system that is highly specific to the plant, and the fact that different viruses can often produce unique patterns of interactions \cite{Gal-On2006,Roossinck2005,Takeshita2005,Bergua2014}. However, if one takes different hypotheses into consideration, depending on the sequence homology of the two viruses and their specificity, one of them could inadvertently trigger an immune response or establish a set of host conditions that could either prevent the secondary infection from taking place or allow it to manifest more aggressively \cite{Malapi-Nelson2009,Reddy2012}. 
 	
 	Current mathematical models of plant virus epidemics with cross-protection have focussed primarily on the transmission dynamics  between populations of healthy plants and plants that are infected with one or multiple viral strains \cite{Zhang2000,Zhang2001,Jeger2011}. By studying the mechanisms of cross-protection on a cellular level, one might achieve a better understanding of the interactions between two viral strains and a single host. In this paper we derive and analyse a model of a plant disease within a single host with particular account for RNAi-mediated cross-protection.  We will show that the model can provide a good qualitative description of the plant's immune response to a viral co-infection, and that it provides a framework in which RNAi can account for both viral synergism and antagonism resulting in cross-protection. A potential application of the model lies in better understanding the efficacy of treating plants against viral diseases by means of the introduction of specific viral strains or genetically modified viruses.
	
	The outline of this paper is as follows. In the next section we describe in detail the main biological assumptions and derive a corresponding mathematical model of plant immune response. In Section 3 we identify all steady states of the model together with conditions for their biological feasibility and stability. Sections 4 is devoted to numerical stability analysis of these steady states, as well as numerical simulations of the model to illustrate different types of dynamical behaviour. The paper concludes with the discussion of results and open problems.
 	\section{Model derivation}
 	 	
 	To investigate the dynamics of biological interactions taking place during a co-infection of a plant with two viruses, we divide the total population of plant cells into the following compartments: healthy (or, susceptible) cells $S(t)$, populations $I_1(t)$ and $I_2(t)$ of cells infectious with virus 1 or virus 2, cells $W_1(t)$ and $W_2(t)$ that are immune to viruses 1 and 2, cells $H_1(t)$ and $H_2(t)$ that have recovered from a primary infection with one of the virus and are currently infectious with the other virus, and finally, the population of super-protected cells $W_{12}(t)$ that are immune to both viruses. Transitions between these different cell populations are illustrated in Fig.~\ref{fig:1}.
 	
 	For the sake of model simplicity, spatial components associated with host-specific anatomy will be neglected, and the cell populations are assumed to uniformly distributed within the plant. Despite potentially overlooking some aspects of the dynamics, the assumption of spatial uniformity has been very effectively used to understand viral dynamics \cite{per02,wod02}. Non-spatial models can provide significant insights into the dynamics and become the basis upon which more detailed models can be built on. Additionally, in the case of field plants, it is biologically reasonable to assume that multiple infection sites could be distributed all over the host. Targeted plants could be exposed multiple times during vector movement or feeding, as vector-borne pathogens have been found capable of even altering the phenotypes of their hosts and vectors in such a way that the frequency and the nature of interactions between them promotes the transmission of the disease \cite{Mauck2010,Moreno-Delafuente2013}. Furthermore, all plant cells are connected through plasmodesmata, the phloem and the xylem vessels responsible for resource translocation \cite{Lalonde2004}, and these pathways can also be used by viruses for systemic infections of their host \cite{Opalka1998,Wan2015}.
 	
 	Plant growth models can generally be divided into two classes: the ones where cell populations are allowed to exhibit unbounded growth, and the ones that assume a certain asymptotic final size due to finite resources or ontogenetic changes, like flowering of the plant. Asymptotic growth models are more favourable in the studies which consider the entire lifespan of the plant \cite{Paine2012,Heinen1999}.
Hence, we will describe plant growth by the logistic growth function with a linear growth factor $r$ and a carrying capacity $K$, with all cell populations contributing to the competition term, as has been effectively done in other models of immune response to infections, such as influenza \cite{tridane}, HIV \cite{PerNel} and HBV \cite{ciupe07}.
	
	Once a plant becomes infected, infected cell populations  $I_1(t)$ and $I_2(t)$ produce new infections by infecting susceptible (healthy) cells at rates $\lambda_1$ and $\lambda_2$, respectively. Due to various metabolic changes and loss of functions that occur after a viral takeover, the lifespan of infected cells is normally shorter than that of healthy cells, as characterised by higher death rates $\epsilon_1$ and $\epsilon_2$. Another possible explanation of a premature death of infected cells is given by the  hypersensitive response of the plant, where infected cells would be programmed to a premature death in order to avoid the spread of the infection and to isolate the infectious site \cite{Zvereva2012,Hinrichs1998,Fritig2007}.
 	
 	In this paper we will assume that a viral infection does not  always have a devastating effect on the cell, and hence it is possible for infected cells to recover before experiencing critical damage. Such recovered cells, denoted by $W_1(t)$ and $W_2(t)$,  will be considered immune to the corresponding viruses in a sense that they are no longer infectious. The recovery rates $\sigma_1$ and $\sigma_2$  represent cumulative effects of the two events mentioned above and represent the rates of transition from infected to warned compartments for each of the two viruses. As described in the Introduction, one of the core mechanisms of the plant immune system is the ability to spread a warning signal that is initiated from infectious sites to other parts of the plant and to protect neighbouring cells against the imminent virus infection.  For the sake of simplicity, the cells that have acquired immunity via this warning signal are also included in $W_1(t)$ and $W_2(t)$ populations. We assume that infected cells initiate and spread the warning signal to healthy cells at the rate $\delta_1$ and $\delta_2$, respectively. Cells that have been the recipients of the propagating signal for both viruses or have recovered from both a primary and a subsequent secondary infection will be represented by the super-protected population of cells $W_{12}(t)$ taken to be immune to both viruses. Thus, warned cells $W_1(t)$ and $W_2(t)$ will be recruited to the super-protected population $W_{12}(t)$ at modified warning rates $\gamma_2\delta_2$ and $\gamma_1\delta_1$, respectively. It is important to note that the resistance to the disease is almost always accompanied by a reduction of fitness normally represented by a reduced reproduction capability of cells \cite{Burdon2003,Tian2003}. In this model we assume no fitness cost in the traditional way, however, immune cells might also experience a shorter lifespan compared to susceptible cells, and, therefore, some fitness cost can be implemented by choosing the appropriate death rate $\epsilon_0$ for super-protected cells $W_{12}(t)$.
 	
 	The warned cells that have acquired immunity to a primary infection but have successfully been infected by a secondary infection will be denoted by $H_i(t)$, where the index $i=1,2$ signifies the current infectious state of the cell. Because of their acquired immunity to one of the viruses, these cells may be less or more resistant to the other virus. If the degree of homology between the two viruses is high, i.e the two viruses are closely immunologically related, it would imply that a cell which is immune or highly resistant to one of the viruses would express the same amount of resistance to both of viruses. On the other hand, if the two viruses are not related, it is reasonable to assume that expressing an antiviral resistance to one of the viruses could induce a susceptibility to a secondary non-related infection by reducing the efficacy of the immune response.
 	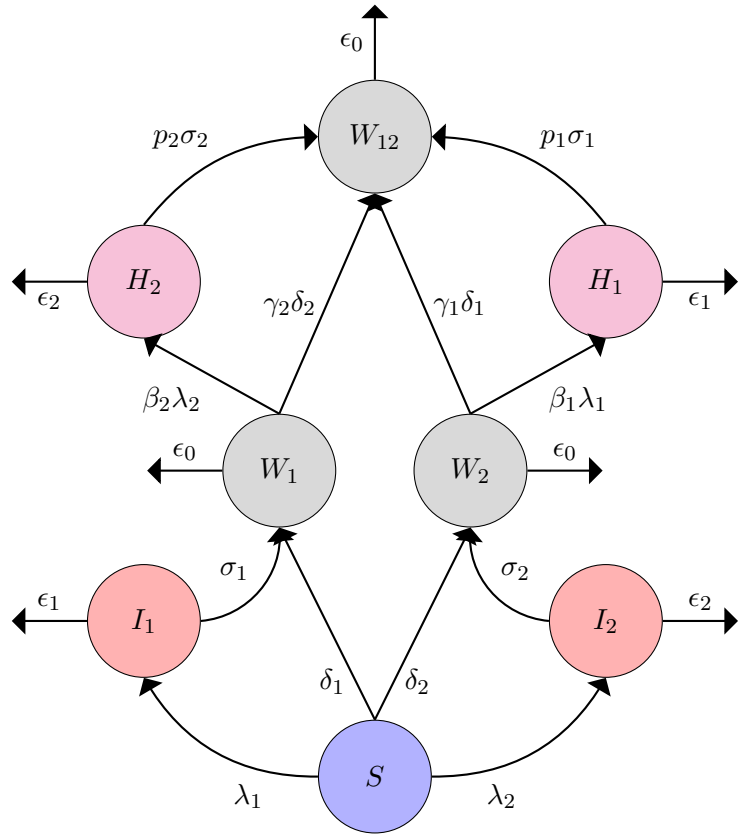
\begin{figure}
 		\centering
 		\begin{tikzpicture}
 		\node (start) [S] {$S$};
 		\node (I1) [I,  above left = 1cm and 2 cm of  start] {$I_1$};
 		\node (I2) [I,  above right =  1cm and 2 cm of start] {$I_2$};
 		\node (H1) [H, above = 3cm of  I2] {$H_1$};
 		\node (H2) [H, above = 3cm of  I1] {$H_2$};
 		\node (W1) [W, above left = 3cm and 0.2 cm of  start]  {$W_1$};
 		\node (W2) [W, above right = 3cm and 0.2 cm of  start] {$W_2$};
 		\node (W12) [W,above = 7 cm of  start] {$W_{12}$};
 		\coordinate[left = 1cm of I1] (LI1);
 		\coordinate[right = 1cm of I2] (RI2);
 		\coordinate[left = 1cm of W1] (LW1);
 		\coordinate[right = 1cm of W2] (RW2);
 		\coordinate[left = 1cm of H2] (LH2);
 		\coordinate[right = 1cm of H1] (RH1);
 		\coordinate[above = 1cm of  W12] (AW12);
 		\draw [arrow,bend left =30] (start.west) to  node[auto,below left, pos = 0.5,pos = 0.2 ]{$\lambda_1$} (I1.south);
 		\draw [arrow,bend right =30] (start.east) to  node[auto, below right, pos = 0.5, pos = 0.2 ]{$\lambda_2$} (I2.south);
 		\draw [arrow] (start.north) to  node[auto,below, pos = 0.2,left] {$\delta_1$} (W1.south);
 		\draw [arrow] (start.north) to  node[auto,below, pos = 0.2, right] {$\delta_2$} (W2.south);
 		\draw [arrow] (W1.north) -- node[anchor=south,left] {$\gamma_2\delta_2$} (W12.south);
 		\draw [arrow] (W2.north) -- node[anchor=south,right] {$\gamma_1\delta_1$} (W12.south);
 		\draw [arrow] (W1.north) to node[auto] {$\beta_2\lambda_2$} (H2.south);
 		\draw [arrow] (W2.north) to node[auto,swap] {$\beta_1\lambda_1$} (H1.south);
 		\draw [arrow, bend right = 45] (I1.east) to  node[auto]{$\sigma_1$} (W1.south);
 		\draw [arrow, bend left = 45] (I2.west) to  node[auto,swap]{$\sigma_2$} (W2.south);
 		\draw [arrow,bend left = 25] (H2.north) to node[auto]{$p_2\sigma_2$}(W12.west);
 		\draw [arrow,bend right = 25] (H1.north) to node[auto,swap]{$p_1\sigma_1$}(W12.east);
 		\draw[arrow](I1.west) to node[swap,auto]{$\epsilon_1$}(LI1);
 		\draw[arrow](I2.east) to node[swap,above]{$\epsilon_2$}(RI2);
 		\draw[arrow](W1.west) to node[swap,above]{$\epsilon_0$}(LW1);
 		\draw[arrow](W2.east) to node[swap,above]{$\epsilon_0$}(RW2);
 		\draw[arrow](W12.north) to node [auto]{$\epsilon_0$}(AW12);
 		\draw[arrow](H2.west) to node[auto]{$\epsilon_2$}(LH2);
 		\draw[arrow](H1.east) to node [auto,swap]{$\epsilon_1$}(RH1);
 		\end{tikzpicture}
 		\caption{A diagram of interactions between two competing viruses and the corresponding plant immune response. Here $S$ denotes the susceptible cells, $I_{1,2}$ and $W_{1,2}$ are the infected and the warned cells for each virus, respectively. Warned cells subsequently infected by a primary or secondary virus are denoted by $H_{1}$ and $H_{2}$. Finally, $W_{12}$ denotes the super-protected cells immune to both viruses. The arrows indicate the rates of transitions from one category of cells to another. } 
 		\label{fig:1}
 	\end{figure}
 	
 	From a biological perspective there could be a limited number of components in the cell that can be used to mount an immune response against a viral infection. For example, unless a cell is warned by both propagating signals, it might be the case that all components able to form antiviral complexes within the cell are being used to prepare only for a single infection, or that there might not be enough components in general to mount a sufficient immune response to both infections simultaneously. Moreover, chemical changes within the cell introduced during the primary infection and the corresponding immune response could potentially provide more favourable conditions in which the secondary infection is established more easily. In light of these observations, the infectious cells $H_1(t)$ and $H_2(t)$ will infect other cells at the modified infection rates $a_1\lambda_1$ and $a_2\lambda_2$  to account for either enhanced ($a_{1,2}>1$) or reduced ($a_{1,2}<1$) viral transmissibility. Similarly, we introduce the susceptibility modifiers $\beta_{1}$ and $\beta_2$  for the warned cells $W_2(t )$ and $W_1(t)$, respectively, which will be assumed to be either susceptible ($\beta_{1,2}>1$) or resistant ($\beta_{1,2}<1$) to the virus agains which they have not yet acquired immunity. To account for a prior infection, the recovery rates of cells $H_i$ are modified by the factors $p_i$, so these cells are recruited into the super-protected population at rates $p_1\sigma_1$ and $p_2\sigma_2$, respectively. Therefore, in this model  the parameters that define viral cooperation will be the modifiers i.e $a_i,\;\beta_i$ and $p_i$ which can be interpreted  as either functions of the antigenic distance or  other specific relation between two viruses. For simplicity, we will ignore  the possibility of random mutations, so that these modifiers will remain constant.
 	
 	Under the above assumptions, the model describing the dynamics of plant immune response to two viral infections can be written as follows,
 	\begin{equation}\label{sys:original}
 	\begin{array}{l}
 	\displaystyle{\frac{dS}{dt}  = r\widehat{S}\left(1-\frac{N}{K} \right)-S\left[(\lambda_1 + \delta_1) I_1 + (\lambda_2 + \delta_2) I_2 + a_2{\lambda}_2H_{2}+a_1{\lambda}_1H_{1}\right],}\\\\
 	\displaystyle{\frac{dI_1}{dt} = I_1(\lambda_1S - \sigma_1 - \epsilon_1)  + a_1{\lambda}_1H_{1}S,}\\\\
 	\displaystyle{\frac{dI_2}{dt} = I_2(\lambda_2S - \sigma_2 - \epsilon_2)  + a_2{\lambda}_2H_{2}S,}\\\\
 	\displaystyle{\frac{d W_1}{d t} = I_1(\sigma_1 + \delta_1S) -W_1\left[\epsilon_0 + (\beta_2 \lambda_2 + \gamma_2\delta_2) I_2 + \beta_2a_2\lambda_2H_{2}\right],}\\\\
 	\displaystyle{\frac{d W_2}{d t} =I_2(\sigma_2 + \delta_2S) - W_2\left[\epsilon_0 + (\beta_1 \lambda_1 + \gamma_1\delta_1) I_1 + \beta_1a_1\lambda_1H_{1}\right],}\\\\
 	\displaystyle{\frac{dH_{1}}{dt} =W_2({\beta}_1{\lambda}_1I_1+\beta_1a_1{\lambda}_1H_{1})- H_{1}(\epsilon_1 + p_1\sigma_1),}\\\\
 	\displaystyle{\frac{dH_{2}}{dt} = W_1({\beta}_2{\lambda}_2I_2+\beta_2a_2{\lambda}_2H_{2})- H_{2}(\epsilon_2 + p_2\sigma_2),}\\\\
 	\displaystyle{\frac{dW_{12}}{dt} =p_1\sigma_1 H_{1} + p_2\sigma_2H_{2} + \gamma_2\delta_2I_2W_1 + \gamma_1\delta_1I_1W_2 - \epsilon_0W_{12},}
 	\end{array}
 	\end{equation}
 	where $\widehat{S}(t)= S(t)+W_1(t)+W_2(t)+W_{12}(t)$, and $N(t)=S(t)+I_1(t)+I_2(t)+W_1(t)+W_2(t)+H_1(t)+H_2(t)+W_{12}(t)$ is the total population of plant cells. As a first step of the analysis, we establish well-posedness of the system (\ref{sys:original}).
 	
 	
 	\begin{theorem}
 		The model (\ref{sys:original}) with initial conditions
			\[
			\begin{array}{l}
			S(0)>0,\quad I_1(0)\geq 0, \quad I_2(0)\geq 0, \quad W_1(0)\geq 0, \quad W_2(0)\geq 0,\\\\
			H_1(0)\geq 0, \quad H_2(0)\geq 0,\quad W_{12}(0)\geq 0,
			\end{array}
			\]
			and $N(0)=N_0 <K$ is well-posed, i.e. its solutions remain non-negative and bounded for all $t\geq 0$.
 	\end{theorem}
 	
 	\begin{proof}
 		Let $T_2$ be a period of time, such that $N(t)<K$ for $t\in[0,T_2]$, and suppose $T_1\le T_2$ is the first time such that $S(T_1) = 0$. This implies that
 			\begin{equation*}
 			\dot{S}(T_1) = r(W_1 + W_2 +W_{12}) [1- (W_1 + W_2 +W_{12} + I_1 + I_2)/K]\ge 0, 
 			\end{equation*}
 			hence,  for any $0\le t\le T_2$, we have that $S(t)\ge 0$. For the remaining variables, considering any positive time $t$, if for any $i=1,2$ we have that $I_i(t) = 0$, this implies that $\dot{I_i}(t)=a_i{\lambda}_iH_{i}\ge 0$, thus $I_i(t) $ must be non-negative for all times. Likewise, for both $W_i(t) = 0$ we obtain $\dot{W_i}(t)=I_i(\sigma_i + \delta_iS)\ge 0$ which shows that $W_i(t)\ge 0$. If $H_i(t) = 0$, we have $\dot{H_i}(t)=W_j{\beta}_i{\lambda}_iI_i\ge 0$ with $1\le i\ne j\le 2$. Finally, for $W_{12}(t) = 0$, we have that $\dot{W}_{12}(t)\ge 0 $. Thus, all variables remain non-negative for  $t\in[0,T_2]$.
 		
 		We now prove, by contradiction, that, in fact, $N(t)<K$ for all $t\ge0 $. Assume, for a contradiction, that there is a first time $T_2>0$ at which the inequality $N(t)<K$ ceases to hold. Since $T_2$ is the first such time, $N(T_2)=K$ and $\dot{N}(T_2) \ge 0 $. As has been shown earlier, all state variables are non-negative at $t=T_2$. Adding up all equations of the system (\ref{sys:original}) yields
 		\begin{equation}
 		\frac{dN}{dt} = r\hat{S}(1-N/K)-\epsilon_1 I_1 - \epsilon_2 I_2 - \epsilon_0 (W_1+W_2) -\epsilon_1 H_1 - \epsilon_2 H_2 - \epsilon_0 W_{12},
 		\end{equation}
 		Since at $t=T_2$ we have that $N(T_2)=K$, the last equation gives $\dot{N}(T_2)<0$, which is a contradiction, unless $I_1(T_2)=I_2(T_2)=W_1(T_2)=W_2(T_2)=H_1(T_2)=H_2(T_2)=W_{12}(T_2)=0$. But in this exceptional case, the initial value theorem for ODEs, applied to the last 7 equations of system (\ref{sys:original}) with $S$ considered as a prescribed function, yields that 
 			$I_1(t)=I_2(t)=W_1(t)=W_2(t)=H_1(t)=H_2(t)=W_{12}(t)=0$ for all $t>T_2$ and the equation for $S(t)$ (the first equation of the system) reduces to the logistic equation $\dot{S} = rS(1-S/K)$. Thus, for any $t\geq T_2$, we have $0<S(t)\leq K$, which completes the proof.
 	\end{proof}
 	
 	To simplify the model and reduce the number of free parameters, we non-dimensionalise the system (\ref{sys:original}) by introducing new dimensionless variables
 	\[
 	\begin{array}{l}
 	\displaystyle{\tau=rt,\quad u_1=\frac{S}{K},\quad u_2=\frac{I_1}{K}, \quad u_3=\frac{I_2}{K},\quad u_4=\frac{W_1}{K},}\\\\
 	\displaystyle{u_5=\frac{W_2}{K}, \quad u_6=\frac{H_1}{K}, \quad u_7=\frac{H_2}{K}, \quad u_8=\frac{W_{12}}{K},}
 	\end{array}
 	\]
 	and for $i= 1,2$, parameters
 	\[
 	L_i = \frac{\lambda_i}{r},\quad d_i = \frac{K\delta_i}{r},\quad e_i = \frac{\epsilon_i}{r},\quad s_i = \frac{\sigma_i}{r}, \qquad e_0 = \frac{\epsilon_0}{r}.
 	\]
 
 	This gives the following modified system
 	\begin{equation}\label{sys:NDS}
 	\begin{array}{l}
 	\displaystyle{\frac{du_1}{d\tau} = \widehat{u}_1(1-\widehat{N}) -u_{{1}} \left[ \left( L_{{1}}+d_
 		{{1}} \right) u_{{2}}+ \left( L_{{2}}+d_{{2}} \right) u_{{3}}+a_{{1}}L
 		_{{1}}u_{{6}}+a_{{2}}L_{{2}}u_{{7}} \right],}\\\\
 	\displaystyle{\frac{du_2}{d\tau}=L_{{1}} \left( a_{{1}}u_{{6}}+u_{{2}} \right) u_{{1}}-u_{{2}} \left( e
 		_{{1}}+s_{{1}} \right),}\\\\
 	\displaystyle{\frac{du_3}{d\tau}= L_{{2}} \left( a_{{2}}u_{{7}}+u_{{3}} \right) u_{{1}}-u_{{3}} \left( e
 		_{{2}}+s_{{2}} \right),}\\\\
 	\displaystyle{\frac{du_4}{d\tau}=u_{{2}} \left( d_{{1}}u_{{1}}+s_{{1}} \right) -u_{{4}} \left[  \left( 
 		\beta_{{2}}L_{{2}}+\gamma_{{2}}d_{{2}} \right) u_{{3}}+\beta_{{2}}a_{{
 				2}}L_{{2}}u_{{7}}+e_{{0}} \right],}\\\\
 	\displaystyle{\frac{du_5}{d\tau}=u_{{3}} \left( d_{{2}}u_{{1}}+s_{{2}} \right) -u_{{5}} \left[\left( 
 		\beta_{{1}}L_{{1}}+\gamma_{{1}} d_{{1}} \right) u_{{2}}+\beta_{{1}}a_{{
 				1}}L_{{1}}u_{{6}}+e_{{0}} \right],}\\\\
 	\displaystyle{\frac{du_6}{d\tau}=\beta_{{1}}L_{{1}} \left( a_{{1}}u_{{6}}+u_{{2}} \right) u_{{5}}-u_{{6
 			}} \left( p_{{1}}s_{{1}}+e_{{1}} \right),}\\\\
 		\displaystyle{\frac{du_7}{d\tau}=\beta_{{2}}L_{{2}} \left( a_{{2}}u_{{7}}+u_{{3}} \right) u_{{4}}-u_{{7
 				}} \left( p_{{2}}s_{{2}}+e_{{2}} \right),}\\\\ 
 			\displaystyle{\frac{du_8}{d\tau} =\gamma_{{1}}d_{{1}}u_{{2}}u_{{5}}+\gamma_{{2}}d_{{2}}u_{{3}}u_{{4}}+p_{{1}}s_{{1}}u_{{6}}+p_{{2}}s_{{2}}u_{{7}}-e_{{0}}u_{{8}},}
 			\end{array}
 			\end{equation}
 			where $\widehat{u}_1 = u_1 +u_4 + u_5 + u_8$ and $\widehat{N} = \widehat{u}_1+ u_2 + u_3 +u_6 + u_7$.
 			\section{Steady states}
 			It is straightforward to see that independently of the values of parameters, the system (\ref{sys:NDS}) always admits a trivial steady state
 			\begin{equation}
 			E_0=(0,0,0,0,0,0,0,0),
 			\end{equation}
 			and a disease-free steady state given by
 			\begin{equation}
 			E_{DF}=(1,0,0,0,0,0,0,0).
 			\end{equation}
 			
 			Looking for steady states of the system (\ref{sys:NDS}) $u_2 = 0$ and $u_{1,3} \ne 0$, gives $u_4 = u_6 =u_7 =u_8 = 0$. Substituting these values in other equations of system (\ref{sys:NDS}) gives a one-virus endemic steady state
 			\begin{equation}\label{E2eq}
 			E_2 =(u_1^*,0,u_3^*,0,u_5^*,0,0),
 			\end{equation}
 			where
 			\[
 			u_1^*=\frac{e_2+s_2}{L_2},\quad u_3^*=\frac{-c_1(u_1^*)-\sqrt{c_1^2(u_1^*)-4c_2(u_1^*)c_0(u_1^*)}}{2c_2(u_1^*)},\quad
 			u_5^*= A(u_1^*)u_3^*,
 			\]
 			with
 			\[
 			\begin{array}{l}
 			\displaystyle{A(u_1^*)=\frac {d_2 u_1^*+s_{2}}{e_0},\quad B=L_2+d_2,\quad c_0(u_1^*)=u_1^*(1-u_1^*),}\\\\
 			\displaystyle{c_1(u_1^*)=A(u_1^*)-u_1^*[2A(u_1^*)+B+1],\quad c_2(u_1^*)=-A(u_1^*)[A(u_1^*)+1].}
 			\end{array}
 			\]
 			The steady state $E_2$ is biologically feasible, as long as the condition $e_2 +s_2<L_2$ holds.
 			
 			Proceeding in a similar manner, one can find a one-virus endemic steady state $E_1$ corresponding to the presence of virus 1 only. This steady state is explicitly given by
 			\begin{equation}
 			E_1=(\widetilde{u}_1^*,u_2^*,0,u_4^*,0,0,0),
 			\end{equation}
 			where now
 			\[
 			\widetilde{u}_1^*=\frac{e_1+s_1}{L_1},\quad u_2^*=\frac{-\widetilde{c}_1(\widetilde{u}_1^*)-\sqrt{\widetilde{c}_1^2(\widetilde{u}_1^*)-4\widetilde{c}_2(\widetilde{u}_1^*)\widetilde{c}_0(u_1^*)}}{2\widetilde{c}_2(u_1^*)},\quad
 			u_4^*= \widetilde{A}(\widetilde{u}_1^*)u_2^*,
 			\]
 			with
 			\[
 			\begin{array}{l}
 			\displaystyle{\widetilde{A}(\widetilde{u}_1^*)=\frac {d_1 \widetilde{u}_1^*+s_{1}}{e_0},\quad \widetilde{B}=L_1+d_1,\quad \widetilde{c}_0(\widetilde{u}_1^*)=\widetilde{u}_1^*(1-\widetilde{u}_1^*),}\\\\
 			\displaystyle{\widetilde{c}_1(\widetilde{u}_1^*)=\widetilde{A}(\widetilde{u}_1^*)-\widetilde{u}_1^*[2\widetilde{A}(\widetilde{u}_1^*)+\widetilde{B}+1],\quad \widetilde{c}_2(\widetilde{u}_1^*)=-\widetilde{A}(\widetilde{u}_1^*)[\widetilde{A}(\widetilde{u}_1^*)+1].}
 			\end{array}
 			\]
 			This steady state is biologically feasible whenever the condition $e_1+s_1<L_1$ is satisfied.
 			
 			Besides the disease-free and the two one-virus endemic steady states, system (\ref{sys:NDS}) can support one or more {\it syndemic} steady states characterised by the simultaneous presence of both viruses,
 			\begin{equation}
 			S= (u_1^*,u_2^*,u_3^*,u_4^*,u_5^*,u_6^*,u_7^*,u_8^*).
 			\end{equation}
 			To find this steady state, let us introduce auxiliary variables and functions
 			\begin{equation}\label{u0def}
 			\begin{array}{l}
 			\displaystyle{u_0= \min_{i=1,2}{\left(\frac{e_i + s_i}{L_i}\right)},\quad F_i(x) = -\frac{\left( L_i x-e_i-s_i\right)}{L_ia_i x},\hspace{0.3cm}i=1,2,}\\\\
 			\displaystyle{\Delta_i(x)=\beta_iL_i(F_i(x)a_i+1)+d_i\gamma_i,\quad G_i(x)= d_i x+s_i,\hspace{0.3cm}i=1,2,}
 			\end{array}
 			\end{equation}
 			which allow us to express all steady state variables through $u_1^*$ in the following way:
 			\[
 			\begin{array}{l}
 			\displaystyle{u_4^*= {\frac {F_2\left(u_1^*\right)\left(p_2s_2+e_2\right)}{\beta_2L_2\left[a_2F_2\left(u_1^*\right) +1 \right] }},\quad
 				u_5^*=\frac {F_1\left(u_1^*\right)\left(p_1s_1+e_1\right)}{\beta_1L_1\left[a_1F_1\left(u_1^*\right)+1\right]},}\\\\
 			\displaystyle{u_2^*= \frac{e_0 u_4^*\left[\Delta_2\left(u_1^*\right) u_5^*+G_2\left(u_1^*\right)\right]}{G_1\left(u_1^*\right)G_2\left(u_1^*\right)-\Delta_1\left(u_1^*\right)\Delta_2\left(u_1^*\right) u_4^*u_5^*},\hspace{0.3cm}
 				u_3^*=\frac{e_0 u_5^*\left[\Delta_1\left(u_1^*\right)u_4^*+G_1\left(u_1^*\right)\right]}{G_1\left(u_1^*\right)G_2\left(u_1^*\right)-\Delta_1\left(u_1^*\right)\Delta_2\left(u_1^*\right)u_4^*u_5^*},}\\\\
 			u_6^*= u_2^*F_1(u_1^*),\quad u_7^*= u_3^* F_2(u_1^*),\\\\
 			\displaystyle{u_8^*= \frac{d_1\gamma_1u_2^*u_5^*+d_2\gamma_2u_3^*u_4^*+p_1s_1u_6^*+p_2s_2u_7^*}{e_0}.}
 			\end{array}
 			\]
 			Substituting these expressions into the equation
 			\[
 			\hat{u}_1^*(1-\hat{N}) -u_1^*\left[ \left(L_1+d_1\right)u_2^*+\left(L_2+d_2\right)u_3^*+a_1L_1u_6^*+a_2L_2u_7^*\right]=0,
 			\]
 			yields a polynomial equation for $u_1^*$, whose roots gives possible candidates for the syndemic steady state. This state is biologically feasible if
 			\[
 			0<u_1^*<u_0,\quad G_1\left(u_1^*\right)G_2\left(u_1^*\right)-\Delta_1\left(u_1^*\right)\Delta_2\left(u_1^*\right)u_4^*u_5^*>0.
 			\]
 			
 			Linearising  system (\ref{sys:NDS}) near the trivial steady state $E_0$ gives the following characteristic equation for eigenvalues $\mu$:
 			\[
 			\left( \mu-1\right)\left(\mu+e_0\right)^3\prod\limits_{i=1}^2(\mu+e_i+s_i)(\mu+p_is_i+e_i)=0.
 			\]
 			Since one of the roots is $\mu = 1$, this implies that the trivial steady state is always unstable, and, therefore, it is impossible for all cell populations to die out. Linearisation near the disease-free steady state $E_{DF}$ has a characteristic equation
 			\begin{equation}\label{eq:chp_Disease_free}
 			\left( \mu+1\right)\left( \mu+e_0\right) ^3\prod\limits_{i=1}^2\left( p_is_i+\mu+e_i\right)\left(\mu-L_i+e_i+s_i\right)=0,
 			\end{equation}
 			implying that the disease-free steady state $E_{DF}$ is linearly asymptotically stable, provided $u_0>1$, with $u_0$ defined in (\ref{u0def}). In epidemiology, one of the most common and efficient techniques for establishing criteria for onset of epidemic outbreaks is analysis of the {\it basic reproduction number} $R_0$, defined as the average number of secondary infections produced by a single infected individual in a totally susceptible population \cite{Hethcote2000,Dietz1993,Driessche2008,Heesterbeek1996}. This quantity can be derived in a number of ways, e.g. using the next generation approach \cite{Driessche2008}, we define the basic reproduction number for each of the viruses as follows
 			 \begin{equation}\label{Def:R0}
 			R_{01} = \frac{L_1}{e_1 + s_1} \qquad R_{02} = \frac{L_2}{e_2 + s_2},	
 			\end{equation}
 			and denote $\displaystyle{R_0 =\max\left\{R_{01},R_{02}\right\}}={u_0}^{-1}$. Then, the disease-free steady state $E_{DF}$ is linearly asymptotically stable if $R_0<1$. This result means that a complete recovery from both viral infections depends on the efficacy of RNA interference from local induction, i.e the ability of the host cell to target and degrade viral RNA in order to inhibit viral multiplication, and also on whether infected cells reach their limited lifespan faster than they can spread the disease for each virus, respectively. Furthermore, since the basic reproduction number $R_0$ does not depend on the transmissibility ($a_{1,2}$) or susceptibility ($\beta_{1,2}$) modifiers, this implies that the interactions between the two viruses during the host co-infection cannot cause both viruses to become extinct. On the other hand, the modifiers may determine whether both viruses, or only one of them will survive.
 			
 			Characteristic equation of linearisation near the endemic steady state $E_2$ can be factorised into 
 			\begin{equation}\label{eq:chp_E2}
 			X_1(\mu) X_2(\mu) X_3 (\mu)=0,
 			\end{equation}
 			where
 			\[
 			\begin{array}{l}
 			\displaystyle{X_1(\mu) = \left(\mu+e_0 \right)\left( p_2s_2+\mu+e_2\right)\left[u_3^*\left(L_2\beta_2+d_2\gamma_2\right)+\mu+e_0\right],}\\\\
 			\displaystyle{X_2(\mu) = \mu^2 + x_{21}\mu + x_{20},\qquad X_3(\mu) = \mu^3 + x_{32}\mu^2 + x_{31}\mu + x_{30},}
 			\end{array}
 			\]
 			and
 			\begin{equation}\label{x21eq}
 			\begin{array}{l}
 			x_{21} = s_1(p_1+ 1 )+2e_1-L_1(a_1\beta_1u_5^* + u_1^*),\\\\
 			x_{20} = (p_1s_1 +e_1) (e_1 +s_1-L_1u_1^*) - L_1a_1\beta_1(e_1 + s_1)u_5^*,\\\\
 			x_{32} = 2u_1^*+\left( L_2+d_2+1\right)u_3^*+2u_5^*+e_0-1,\\\\
 			x_{31} = d_1(u_3^*)^2 + \left[(L_2 + d_2)[u_1^*(L_2+1)+u_5+e_0]+d_2(u_1^*+u_5^*-1)+ e_0\right]u_3^*\\
 			\qquad+ e_0(2u_1^*+2u_5^*-1),\\\\
 			x_{30} = L_2u_3^*\left[d_2u_1^*\left(2(u_1^*+u_5^*)+u_3^*+e_0-1\right)+u_1^*e_0\left(L_2+1\right)+s_2\left(2u_1^*+u_3^*-1\right)\right] \\
 			\qquad+L_2u_3^*+u_5^*\left(e_0+2s_2\right).
 			\end{array} 
 			\end{equation}
 			Since all system parameters are strictly positive, the roots of $X_1(\mu)$ are all real and negative. By the Routh-Hurwitz criterion we have that all roots of $X_2(\mu)$ lie in the left complex half-plane if the coefficients $x_{21}$ and $x_{20}$ are positive, which translates into the requirements
 			\begin{equation}
 			\begin{aligned}[alignment]
 			&u_5^*<\frac{(s_1p_1 +e_1)+(s_1 + e_1 -L_1u_1^*)}{L_1a_1\beta_1}:=u_A,\quad \mbox{and}\\
 			&u_5^*<\frac{(s_1p_1 +e_1)(s_1 + e_1 -L_1u_1^*)}{L_1a_1\beta_1(e_1 + s_1)}:=u_B.
 			\end{aligned}
 			\end{equation}
 			Since $u_5^*$ must be positive, we require that $u_1^*< (s_1 + e_1)/L_1$. Additionally, since $s_1 + e_1 -L_1u_1^*<s_1 + e_1$, we have
 			\begin{equation}\label{eq:u_B}
 			\begin{aligned}
 			u_A =& \frac{s_1p_1 +e_1}{L_1a_1\beta_1} + \frac{s_1 + e_1 -L_1u_1^*}{L_1a_1\beta_1}>\frac{s_1p_1 +e_1}{L_1a_1\beta_1},\\
 			u_B =& \frac{s_1p_1 +e_1}{L_1a_1\beta_1}\frac{s_1 + e_1 -L_1u_1^*}{e_1 + s_1}<\frac{s_1p_1 +e_1}{L_1a_1\beta_1},
 			\end{aligned}
 			\end{equation}
 			implying $u_B<u_A$. Hence, the roots of $X_2(\mu)$ have a negative real part, provided
 			\[
 			u_1^*<\frac{s_1 + e_1}{L_1} =\widetilde{u}_1^* \quad\mbox{and}\quad u_5^*<u_B.
 			\]
 			This also implies that a necessary condition for the stability of the endemic steady state $E_2$ is the intuitively natural result that the two basic reproduction numbers defined in (\ref{Def:R0}) must satisfy $R_{02}>R_{01}$.
 		
 			\begin{table}
 				\begin{tabular}{|  p{.20\linewidth}|   p{.52\linewidth}| p{.18\linewidth}|}
 					\hline
 						Dimensionless Parameters & Biological meaning & Baseline value\\
 					\hline 
 					$\quad L_{1,2}$  & Infection rate & 1.5 \\
 					$\quad s_{1,2}$  & Recovery rate & 0.5\\
 					$\quad d_{1,2}$  & Propagation rate & 0.05\\
 					$\quad a_{1,2}$     & Transmissibility modifier (after secondary infection) & 1  \\
 					$\quad\beta_{1,2}$   & Susceptibility modifier (after primary infection) & 1 \\
 					$\quad \gamma_{1,2}$   & Acquired secondary immunity modifier & 0.5 \\
 					$\quad e_0$        & Natural death rate  & 0.3\\
 					$\quad e_{1,2}$ & Infected cell death rate & 0.6 \\
 					$\quad p_{1,2}$   & Recovery modifier& 0.2 \\
 					\hline
 				\end{tabular}
 				\caption{Table of  baseline values of parameters in system (\ref{sys:NDS}).}
 				\label{tab:List of parameters}
 			\end{table}

 			Applying the Routh-Hurwitz criterion to the cubic polynomial $X_3(\mu)$ gives that all roots of this polynomial have negative real parts, as long as $x_{32}, x_{31}$ and $x_{30}$ are positive and satisfy the condition $x_{32}x_{31}>x_{30}$. It is important to note that stability of the endemic steady state $E_2$ does not depend on the susceptibility  and transmissibility modifiers $a_2$ and $\beta_2$. From a biological perspective, this suggests that the capability of the second virus to survive as a single infection is irrelevant from the point of view of its ability to infect cells that are chemically altered and are immune to the first virus. On the other hand, the ability of viruses $d_i$ to trigger a warning signal appears to control whether they can exclude each other or co-exist in a stable equilibrium. Hence, we have proved the following result.
 			\begin{theorem}
 				For the endemic steady state $E_2=(u_1^*,0,u_3^*,0,u_5^*,0,0)$ with $u_1^*=(e_2+s_2)/L_2$, $u_3^*$ and $u_5^*$ given in (\ref{E2eq}), let $x_{30}$, $x_{31}$, $x_{32}$ and $u_B$ be defined by (\ref{x21eq}) and (\ref{eq:u_B}), respectively. Then the steady state $E_2$ is linearly asymptotically stable if and only if the following conditions hold.
 				\begin{enumerate}[(i)]
 					\item $0<u_5^*<u_B,$
 					\item $x_{30}>0$, $x_{31}>0$, $x_{32}>0,$
 					\item $x_{32}x_{31}>x_{30}.$
 				\end{enumerate}
 			\end{theorem}
 			
 			\begin{figure}
 				\hspace{-0.5cm}
 				\epsfig{file=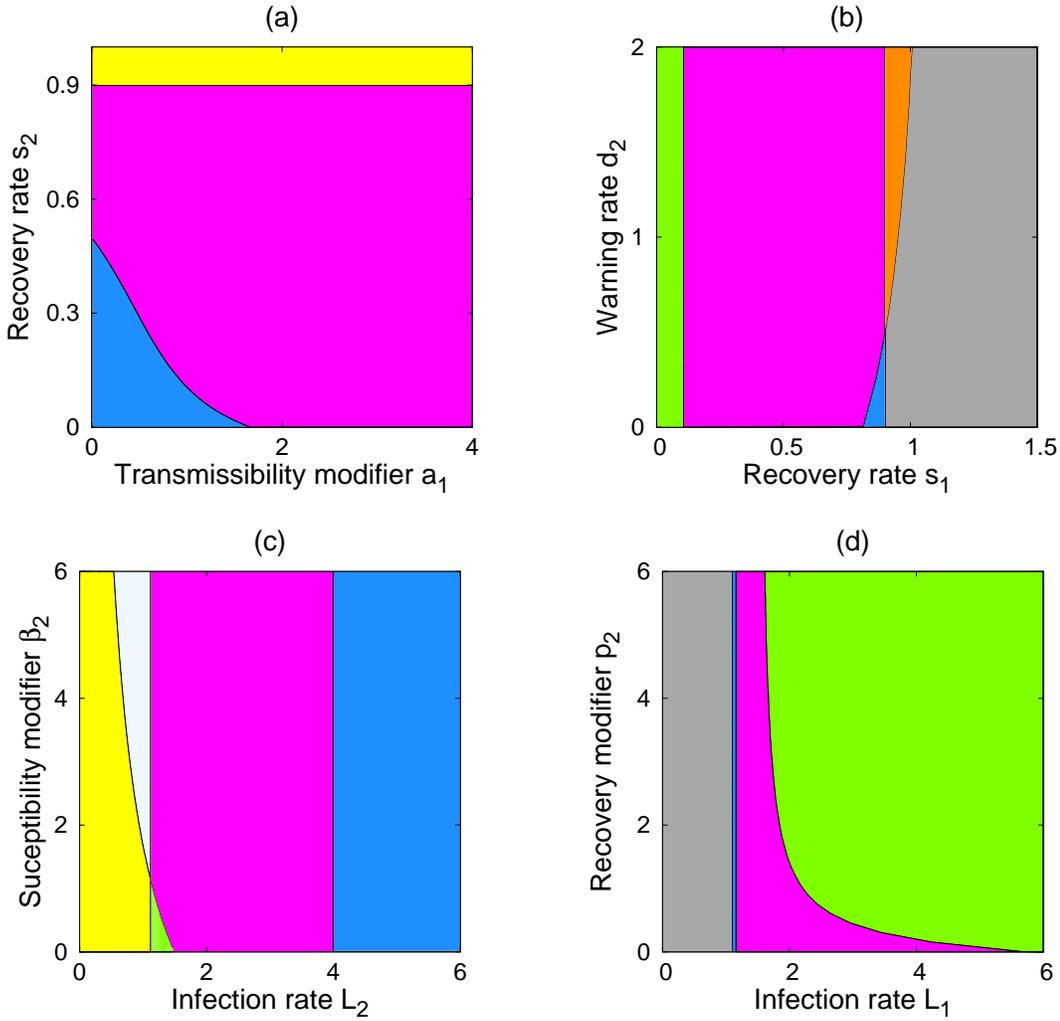,width=14cm}	
 				\caption{Stability of steady states of the system (\ref{sys:NDS}) with parameters from Table~\ref{tab:List of parameters}. Green and blue indicate regions where both endemic steady states $E_1$ and $E_2$ are feasible, but only $E_1$ or $E_2$ is stable, respectively. Magenta shows the region where all three infected steady states are feasible, but only the syndemic steady state $S$ is stable. Yellow is the area where only $E_1$ is feasible and stable, whereas grey is the area where only $E_2$ is feasible and stable. White and orange is where the syndemic steady state is stable, whereas $E_1$ or $E_2$, respectively, is also feasible (and unstable).}
 				\label{fig:2}
 			\end{figure}
 			
 			\begin{remark}
 				The result of this {\bf Theorem} can be applied to the analysis of stability of the endemic steady state $E_1$ by swapping parameter indices $1$ with $2$, and replacing variables $u_3^*$ and $u_5^*$ with $u_2^*$ and $u_4^*$, respectively, as a consequence of the model symmetry.Unlike some other models of multi-strain/multi-virus infections \cite{Gupta1999,Allen2003,Castillo-Chavez1996}, the complexity of the model (\ref{sys:NDS}) prevents one from expressing the conditions for stability of single-virus or co-existence equilibria in a closed form depending only on two basic reproduction numbers.
 			\end{remark}
 			
 			Since the syndemic steady state $S$ cannot be found in a closed form, it does not prove possible to derive analytical conditions for stability of this steady state. Hence, to understand how stability changes with parameters,  one has to resort to numerically compute eigenvalues of the Jacobian of the linearisation of system (\ref{sys:NDS}) near the steady state .
 			
 			\section{Numerical stability analysis and simulations}
 			
 			 Due to RNAi being a very complicated multi-component process, obtaining accurate parameters values to be used in a mathematical model is extremely difficult and often impractical, as  some parameters cannot currently be measured, or even when they do, there is a very wide variability in the reported values \cite{Melnyk2011,Liang2012, Himber2015}. Parameter values that define viral properties and modifiers in the context of this study are equally problematic to obtain, as one would require virus-specific information about both the cell-to-cell and long-distance transmission of the virus. For example, in the case of the Tobacco mosaic virus, the infection can on average spread from one cell to another every 3-4 hours depending on the strain of the virus and the temperature \cite{Kawakami2004}, and although this information can provide some intuition about parameter values, it is not sufficient for estimating the actual infection rate.
 			 
 			To better understand the effects of different parameters on feasibility and stability of different steady states of the system (\ref{sys:NDS}), we use {\bf Theorem 3.1} and numerical computation of eigenvalues to identify parameter regions associated with existence and stability of all steady states. To this end, we start with baseline parameter values given in Table~\ref{tab:List of parameters} and allow some of the parameters to vary. Since model (\ref{sys:NDS}) has quite a large number of different parameters, below we present the results for only some parameter combinations that illustrate the diversity of possible scenarios, and qualitatively similar results can be obtained when other parameters are varied. Plotting the percentages of infected cells for each steady state in the same parameter space allows us to investigate possible changes in the magnitude of the infected cell population between different steady states.
 			
 			Figures \ref{fig:2}, \ref{fig:3} and \ref{fig:4} illustrate earlier analytical conclusions that the two endemic steady states $E_1$ and $E_2$ are only feasible and stable if the recovery/death rates of infected cells are sufficiently low. On the other hand, one expects that a virus can only survive if its infection rate is adequately high, as observed in Fig.~\ref{fig:2}(d) and Fig.~\ref{fig:4}(a). If either one of the recovery/infection rates is below or above a certain threshold, it is easy to see that the syndemic steady state disappears, and only one of the two viruses survives. However, Figures~\ref{fig:2}(a), (c)  and \ref{fig:3}(c), together with additional computations not shown here, suggest that by increasing parameters $a_{1,2}$, i.e the transmissibility modifiers, or the susceptibility modifiers $b_{1,2}$, the system can generally move from one of the endemic steady states to a stable syndemic equilibrium. This suggests that the most competitive viral strain, which under different circumstances would be capable of excluding a secondary infection, might instead  facilitate the survival of a secondary strain. Cells that have been chemically  altered by the immune response to the more aggressive strain can now serve as ideal targets in which the second strain could proliferate. Since for the fixed values of other parameters, infection rates $L_1$ and $L_2$ are proportional to the two basic reproduction numbers, $R_{01}$ and $R_{02}$, respectively, Figure~\ref{fig:4}(a) effectively is equivalent to figures demonstrating the dependence of steady states on basic reproduction numbers in two-strain models of infectious diseases \cite{Gupta1996,Andreasen1997}.
			
 			\begin{figure}
 				\hspace{-0.5cm}
 				\epsfig{file=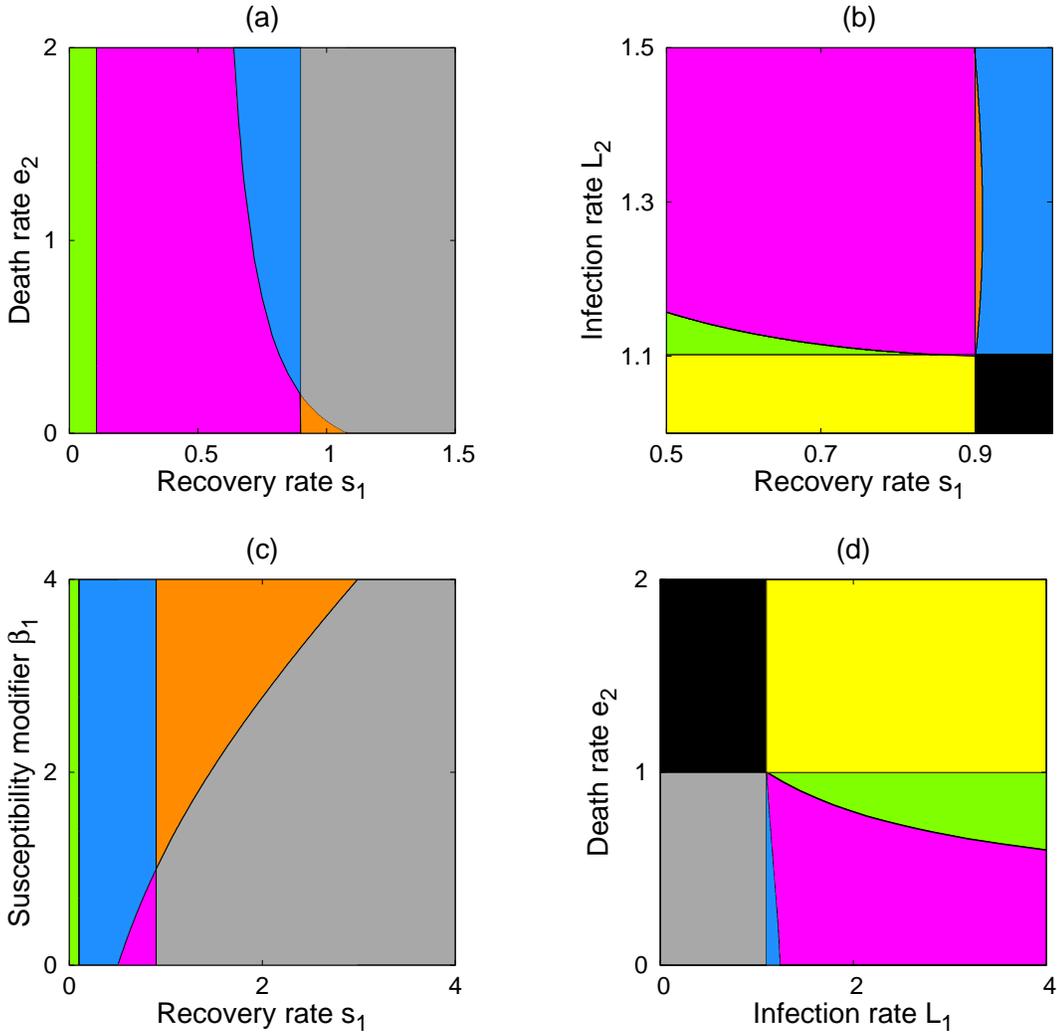,width=14cm}
 				\caption{Stability of the steady states of the system (\ref{sys:NDS}) with parameters from Table~\ref{tab:List of parameters}. Green and blue indicate regions where both endemic steady states $E_1$ and $E_2$ are feasible, but only $E_1$ or $E_2$ is stable, respectively. Magenta shows the region where all three infected steady states are feasible, but only the syndemic steady state $S$ is stable. Yellow is the area where only $E_1$ is feasible and stable, whereas grey is the area where only $E_2$ is feasible and stable. Orange is where the syndemic steady state is stable, and $E_2$ is feasible but unstable. Black is the region where only the disease-free steady state is feasible and stable.}
 				\label{fig:3}
 			\end{figure}
 			
 			Figures~\ref{fig:2}(d) and \ref{fig:3}(a) show that when one of the recovery modifiers  $p_{1,2}$ is increased, the system can move from the syndemic to one of the endemic equilibria $E_{1,2}$, thus behaving in a  qualitatively opposite way to an increase of the corresponding  parameter pair $\left\{a_i,\beta_i\right\}$. This occurs when cells with acquired immunity to one of the viruses are subsequently infected with another virus but have a faster recovery. As this reduces the overall spread of the secondary infection, it will inevitably allow the primary virus to dominate and eventually be the sole survivor in the host. In  Fig.~\ref{fig:2}(b) one observes that by increasing the dimensionless warning rate $d_2$ we can move from a parameter region where only the endemic steady state $E_2$ is feasible and stable (a grey region) to a region, where the syndemic  equilibrium is also stable (an orange region). This suggests that the plant immune response to the second virus can establish conditions that are more favourable to the first virus. Thus, in the case of a double infection, it is possible for a viral infection to persevere in the presence of the host's immune response despite being unable to do so as a single infection. This means that the propagating component of the immune response plays a significant role in the interactions between two viruses and can dictate whether both of them can survive in a single host.
			
			\begin{figure}
 				\hspace{-0.5cm}
 				\epsfig{file=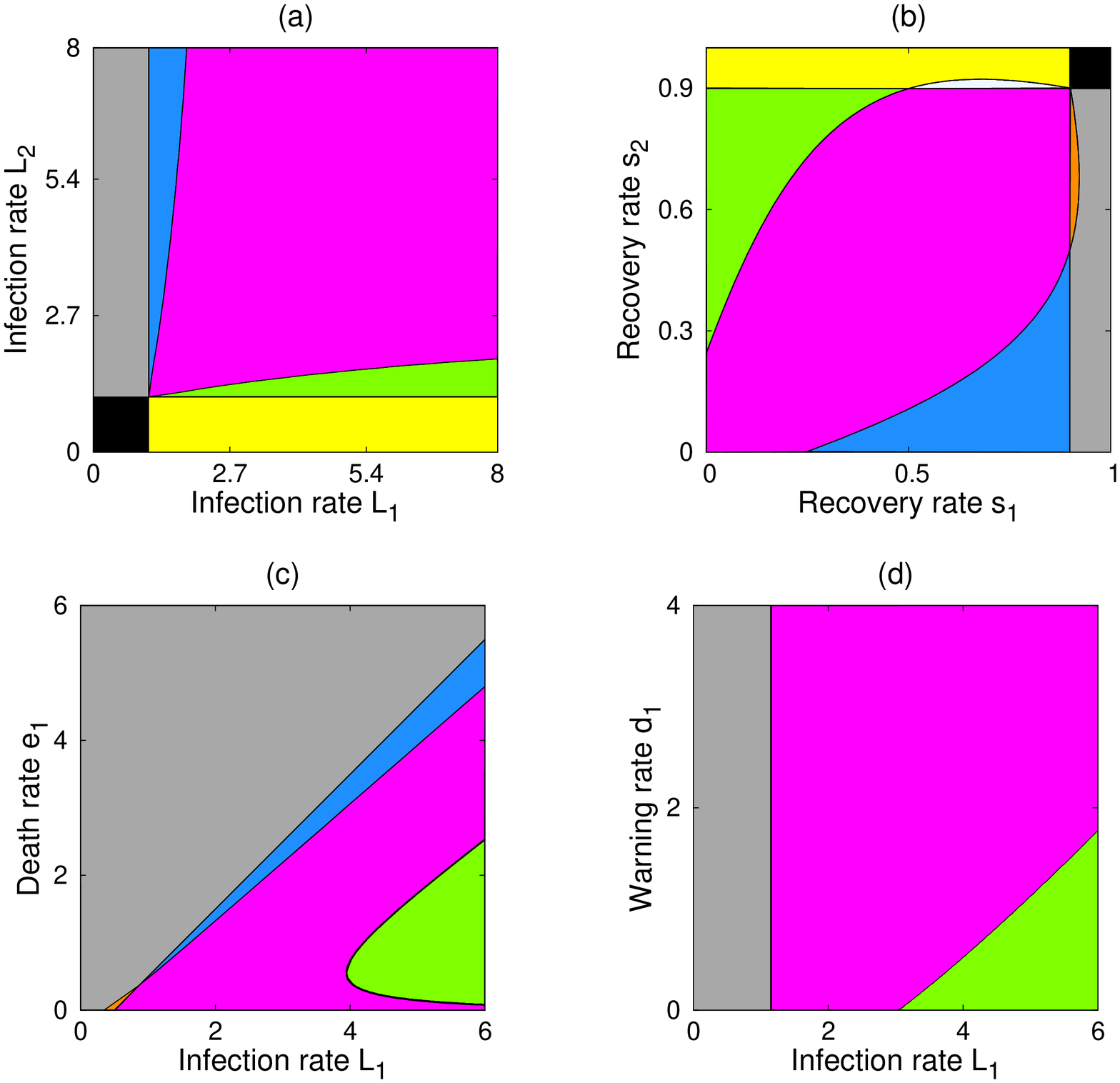,width=14cm}
 				\caption{Stability of steady states of the system (\ref{sys:NDS}) with parameters from Table~\ref{tab:List of parameters}. Green and blue indicate regions where both endemic steady states $E_1$ and $E_2$ are feasible, but only $E_1$ or $E_2$ is stable, respectively. Magenta shows the region where all three infected steady states are feasible, but only the syndemic steady state $S$ is stable. Yellow is the area where only $E_1$ is feasible and stable, whereas grey is the area where only $E_2$ is feasible and stable. White and orange is where the syndemic steady state is stable, whereas $E_1$ or $E_2$, respectively, is also feasible (and unstable). Black is the region where only the disease-free steady state is feasible and stable.}
 				\label{fig:4}
 			\end{figure}	
			
 				Recall that in model (\ref{sys:NDS}), the two viruses are considered to cooperate with each other when $a_i,\beta_i>1$, have a neutral relationship when $a_i,\beta_i = 1$, and ``antagonize" each other when $a_i,\beta_i<1$, $i=1,2$. One should also note the existence of other more complicated scenarios as each of $a_1$, $a_2$, $\beta_1$ and $\beta_2$ can be less than, greater than or equal to one. For example, if $a_1,\beta_1 > 1$ and $0<a_2,\beta_2<1$ then, cooperation of the two viruses will be considered to benefit mostly the first virus, thus being unequal. On the other hand, for $a_2,\beta_2 > 1$ and $a_1,\beta_1 = 0$, the relationship is completely one-sided in favour of the second virus. Figures~\ref{fig:5}(a) and (b) suggest that the biological interactions between different viruses may sometimes disproportionately favour one of the viruses and decrease the potency of the second infection, that is to say that one of the viruses experiences less spread during a co-infection when compared to its single-virus infected steady state. This is clearly evident in Fig.~(\ref{fig:5})(b): for $\beta_2\le 0.87$ only the first virus is present, whereas for  $\beta_2 > 0.87$ the system moves into the syndemic steady state where now both viruses are able to survive, but the first virus is not as widely spread as before. One should note that this result comes at the cost of increasing the total number of infected cells, suggesting that it might not always be the preferable outcome for the plant. Similarly, Figure~(\ref{fig:5})(a) shows that for small values of the transmissibility modifier $a_1$ combined with a higher infection rate $L_2>L_1$ (which also implies $R_{02}>R_{01}$), only the second virus is able to survive in the host. As the value of $a_1$ increases, the picture changes, and the system moves to a syndemic steady state, where not only both of the viruses are able to survive, but given sufficiently high value of $a_1$, the first virus can become dominant. This also suggests that increasing $a_i$ is qualitatively interchangeable increasing $\beta_j$ for $j\neq i$.  Figures~\ref{fig:5}(c) and (d) show how depending on the level of cooperation between the two viruses, i.e for sufficiently high values of $a_1$ and $a_2$,  it can be beneficial for the viruses to co-exist, as they can both infect a bigger biomass of the host compared to their respective one-virus steady states, possibly resulting in a chronic condition that is more severe. These results show that sufficient levels of mutual cooperation between two viruses promote their virulence and ensure that neither of them becomes eradicated, which eventually leads to a persistent double infection with parameter values determining the magnitude of each infection.
				
 			\begin{figure}
				\epsfig{file=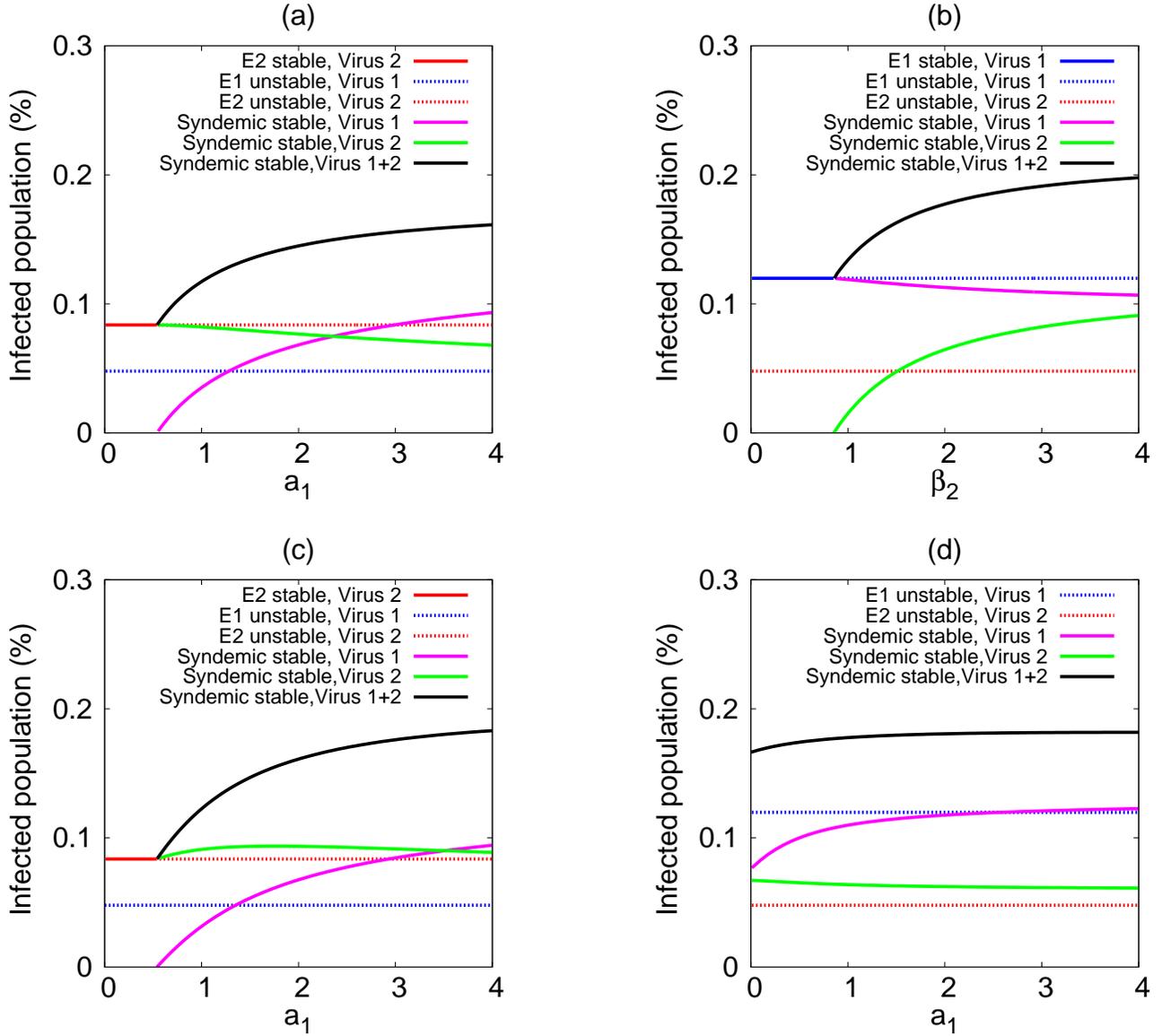,width=\linewidth}
 				\caption{Stability of endemic and syndemic steady states of the system (\ref{sys:NDS}) with parameter values from Table~\ref{tab:List of parameters}. Stable (unstable) steady states are indicated by solid (dotted) lines for single-virus endemic steady states $E_1$ (blue) and $E_2$ (red). The percentage of cells at the syndemic steady state is illustrated for virus 1 (magenta), virus 2 (green), and the total infected population (black). (a) $L_1<L_2 = 2$. (b) $L_2<L_1 = 3.$ (c) $L_1<L_2 = 2$ and $a_2 = 2$. (d) $L_2<L_1 = 3$, $a_2 = 2$ and $\beta_1 = 1.5$.}
 				\label{fig:5}
 			\end{figure}
				
	If the cooperation between the two viruses is unequal or one-sided, it is possible that the least benefited virus will experience less spread compared to its single-virus infected steady state. To investigate scenarios where both viruses "antagonize" each other we solved the system at $a_i,\beta_i =0.5$, $i=1,2$. One result is given in Fig.~(\ref{fig:6})(a), and it shows that increasing $\beta_2$ decreases the presence of the first virus, but similarly to our previous results it increases the overall level of infection. The most interesting case is shown in Fig.~(\ref{fig:6})(b), where  adequately increasing the warning rate $d_2$, not only the percentage of cells infected with the second virus goes down, but also the total number of infected cells is reduced. One also observes in this Figure that although the number of cells infected with the first virus is slowly increasing, it is still at a much lower level than what it was in the absence of the second virus, i.e compared to the steady state $E_1$ which is now unstable. This situation represents an ideal scenario, where inoculating the target plant with a less harmful virus or viral strain can offer partial protection against another specific virus or strain, thus potentially minimizing damage to the host.
								
 			To demonstrate different kinds of dynamics that can be exhibited by the model, we have solved the system (\ref{sys:NDS}) numerically for different combinations of parameters, and the results are shown in Fig.~\ref{fig:num_sim}.  Figure~\ref{fig:num_sim}(a) shows the solution of the model that approaches the stable syndemic steady state, with all compartments having positive values.  As we discussed earlier, from a biological perspective this represents the cases where interactions between the two viruses  facilitate the survival of both viral species within the same host. Figures~\ref{fig:num_sim}(b) and (c) illustrate situations where one of the viruses survives, while the other one is eradicated by the plant immune system, and Figure~\ref{fig:num_sim}(d) demonstrates the case where the plant makes a full recovery.

 			\section{Discussion}
 			
 			In this paper we have derived and analysed a mathematical model of biological interactions between two viruses and a single plant host, with particular account for RNA interference.
			
			\begin{figure}
				\epsfig{file=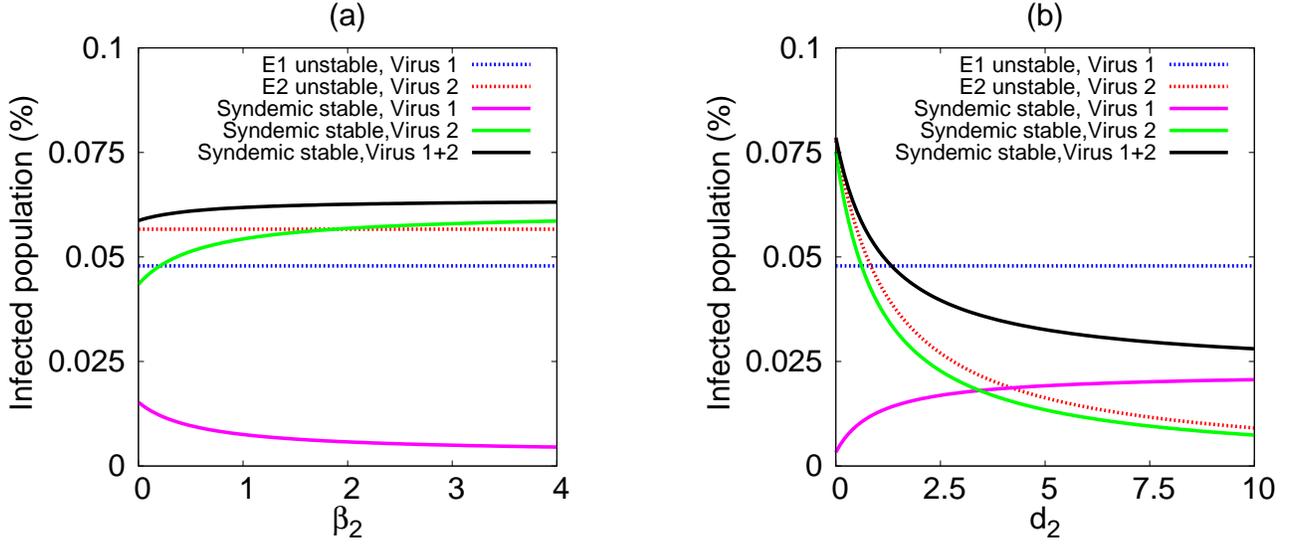,width=\linewidth}
 				\caption{Stability of endemic and syndemic steady states of the system (\ref{sys:NDS}) with  $L_1<L_2 = 1.6$ and the other parameter values given in Table~\ref{tab:List of parameters}. (a) $a_{1,2}=\beta_1=0.5$. (b) $a_{1,2}=\beta_{1,2}=0.5$ . Stable (unstable) steady states are indicated by solid (dotted) lines for single-virus endemic steady states $E_1$ (blue) and $E_2$ (red). The percentage of cells at the syndemic steady state is illustrated for virus 1 (magenta),  virus 2 (green), and the total infected population (black).}
 				\label{fig:6}
 			\end{figure}
			
			Our results have shows that RNA interference can provide a mechanism for cross-protection, and a co-infection can either increase or decrease the overall potency of individual infections, illustrating how cross-protection or cross-enhancement can occur between the two viruses. The framework developed in this paper can be directly applicable to analysis of RNAi-mediated interactions for many combinations of  plant viruses, with examples including co-infections with Soybean mosaic virus and Alfalfa mosaic virus \cite{Malapi-Nelson2009}, as well as Abutilon mosaic virus and Cucumber mosaic virus \cite{Wege2007}. The model can also be used to obtain insights into how one could control viral diseases through cross-protection and, by extension, through gene and antiviral therapy, where genetically modified viruses are introduced to the host. Unlike the wild type strains, these modified viruses can be engineered to deliver specific therapeutic siRNA, which through the process of RNA interference would trigger immune response, thus acting as a powerful vaccination strategy \cite{Silva2002,Caplen2004,Soifer2007}.
 			
 			To achieve greater biological realism, we have assumed that the new plant growth depends on the availability of healthy cells which can be impeded once the plant becomes infected. Stability analysis of the steady states has demonstrated the significance of different parameters of the model and showed how they dictate the dynamical behaviour exhibited by the system.
			
			\begin{figure}
 				\centering
 				\epsfig{file=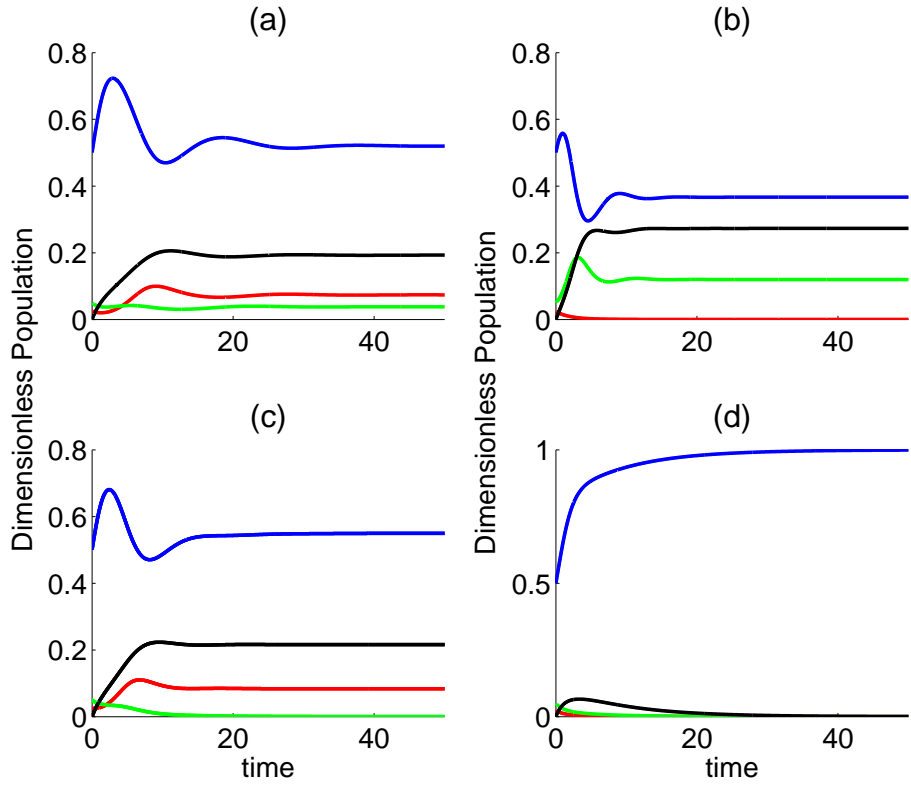,width=14cm}
 				\caption{Numerical simulations of the model (\ref{sys:NDS}) with parameter values from Table~\ref{tab:List of parameters}. Colours represent dimensionless populations of susceptible cells (blue), cells infected with the first (red) and second (green) virus, the total population of cells with immunity to one or both viruses (black). (a) Stable co-existence of viruses: $a_1=3$. (b) Stable single-virus state $E_2$:  $L_1 = 1$ and $L_2 = 3$. (c) Stable single-virus state $E_1$: $L_1 = 2$ and $p_2 = 3$. (d) Stable disease-free steady state $E_{DF}$: $s_1 = 1$ and $L_2 = 1$.}
 				\label{fig:num_sim}
 			\end{figure}
			
			One should note that in the current model it is impossible for all cell populations to die, as there will always be some new growth taking place to replace the parts of the plant that are lost either naturally or due to infections. This is true despite the growth penalty introduced by allocating some of the resources to infected parts of the plant.  Even if a plant were to experience a severe case of stunting, it would be highly unlikely that every healthy cell would become infected, and, therefore, lead to the death of the plant. Although our model cannot capture this scenario, realistically, such events do occur quite rarely in nature.
 			
 			Stability of the disease-free equilibrium and the feasibility of the two single-virus endemic steady states depends on the two basic reproduction numbers $R_{01}$ and $R_{02}$. In our model these quantities are represented as functions of only the infection, recovery and death rates of infected cells for each  strain, but they are not affected by the propagating component of the immune response. This suggests that a faster mobile signal can at most help the plant to recover faster (as determined by the above-mentioned factors) but, by itself it is not sufficient for a recovery. However, this picture changes when the stability of the syndemic steady state is considered. Our results show that the warning signal plays a significant role in determining whether both viruses can persist simultaneously, and as such, it controls situations where the plant is able to support some constant level of both infections. If the two viruses are sufficiently immunologically related, so that the immune response trigged by a primary infection with a less virulent strain could induce a sufficient response against a secondary strain, then the least harmful of the two viruses becomes dominant, and the plant experiences a degree of cross-protection which may sometimes result in the increased total population of infected cells.
 				
 			Analysis of the model has demonstrated that the total population of infected cells during a co-infection can sometimes, but not always, be higher than during a single infection, for which there are two possible explanations. One possibility is that the two different infections simply increase the overall rate of infection. Another aspect is that the two viruses only have to compete for susceptible cells, as there is a source of cells that might be exclusively available to each of the viruses, i.e cells that have acquired immunity to one virus may be less or more susceptible to the other virus. Our results have shown that when two viruses ``antagonize" each other, i.e  $a_i,\beta_i<1$, for sufficiently high warning rates, not only can one minimize the spread of a specific virus, but the overall infection can also reduce. Hence, depending on the virulence of the two strains, one might choose to either avoid the introduction of a secondary viruses, or instead use it in order to produce the more favourable outcome.
 				
 			If the two viruses are immunologically unrelated and co-infecting the same plant, they can indirectly promote each other by inadvertently making cells they can no longer infect more susceptible to the other virus. Hence, despite the fact that both viruses are effectively competing for the same resource, there is always some exclusive source of potential cells in which the infection could survive, with the potency of individual infections strongly dependent on the interaction between the two viruses. Another important  result  is that the syndemic steady state can potentially be stable in parameter regions where only one of the endemic steady states is feasible, implying that a secondary virus can only survive when another infection is present. 
 			
 			One possible extension of the work presented in this paper is to explicitly include in the model time delays associated with plant maturation time, and with delayed propagation of the RNAi signal, as has been recently done in a simpler model of immune response to a single viral infection \cite{Neofytou2016}. Another interesting phenomenon to consider is the possibility of cells being occupied by two viruses simultaneously, which would allow for a wider spectrum of interactions between the viruses and their host. This could include super-infection of individual cells, viral interference or recombination events that can give rise to additional strains.
 			           
			           \section*{Acknowledgements} The authors would like to thank two anonymous reviewers for
helpful comments and suggestions that have helped to improve the presentation in this paper.


\begin{thebibliography}{100}
           		%
           		\bibitem{Kuntz2014} Kuntz, M., 2014. Is it possible to overcome the GMO controversy? Some elements for a philosophical perspective, in Ricroch, A., Chopra, S., Fleischer, S.J. (Eds.), Plant Biotechnology, Springer, Berlin, 107--111.
           		
           		\bibitem{Savary2012} Savary, S., Ficke, A., Aubertot, J.-N., Hollier, C., 2012. Crop losses due to diseases and their implications for global food production losses and food security. Food Sec. 4, 519--537.
           		
           		\bibitem{Purcell2005} Purcell, A.H., Almeida, R.P.P., 2005. Insects as vectors of disease agents, in Dekker, M. (Ed.), Encyclopedia of plant and crop science. Taylor and Francis, London.
           		
           		\bibitem{Chan1994} Chan, M.S., Jeger, M.J., 1994. An analytical model of plant virus disease dynamics with roguing and replanting. J. Appl. Ecol. 31, 413--427.
           		
           		\bibitem{VandenBosch1996} van den Bosch, F., de Roos, A.M., 1996. The dynamics of infectious diseases in orchards with roguing and replanting as control strategy. J. Math. Biol. 35, 129--157.
           		
           		\bibitem{Zhang2012a} Zhang, T., Meng, X., Song, Y.,  Li, Z., 2012. Dynamical analysis of delayed plant disease models with continuous or impulsive cultural control strategies. Abstr. Appl. Anal. 2012, 1--25.
           		
           		\bibitem{Gutierrez1974}Gutierrez, A.P., Havenstein, D.E., Nix, H.A., Moore, P.A., 1974. The ecology of Aphis craccivora Koch and Subterranean Clover Stunt Virus in South-East Australia. II. A Model of Cowpea Aphid Populations in Temperate Pastures. J. Appl. Ecol. 11, 1--20. 
           		
           		\bibitem{Frazer1977}Frazer, B.D. 1977. Plant virus epidemiology and computer simulation of aphid populations. Aphids as Virus Vectors (ed.by K. F. Harris and K. Maramorosch), pp. 413-431. Academic Press, New York.
           		
           		\bibitem{Kirtani1978}Kirtani, K., Sasaba, T., 1978. An experimental validation of the systems model for the prediction of rice dwarf virus infection. Appl. Entomol. Zool. 13, 209--214.
           		
           		\bibitem{Irwin2000}Irwin, M.E., Ruesink, W.G., Isard, S.A., Kampmeier, G.E., 2000. Mitigating epidemics caused by non-persistently transmitted aphid-borne viruses: the role of the pliant environment. Virus Res. 71, 185--211.
           		
           		\bibitem{Madden2000} Madden, L.V., Jeger, M.J., van den Bosch, F. (2000) A theoretical assessment of the effects of vector-virus transmission mechanism on plant virus disease epidemics. Phytopathology, 90, 576--594.
           		
           		\bibitem{Zhang2000} Zhang, X.S., Holt, J., Colvin, J., 2000. Mathematical models of host plant infection by helper-dependent virus complexes: why are helper viruses always avirulent? Phytopathology 90, 85--93. 
           		
           		\bibitem{Frank2002}Frank, S.A., 2002. Immunology and Evolution of Infectious Disease. Princeton University Press, Princeton.
           		
           		\bibitem{Lipsitch2007}Lipsitch, M., O'Hagan, J.J., 2007. Patterns of antigenic diversity and the mechanisms that maintain them. J. Roy. Soc. Interface 4, 787--802. 
           		
           		\bibitem{Gupta1994}Gupta, S., Trenholme, K., Anderson, R.M., Day, K.P., 1994. Antigenic diversity and the transmission dynamics of Plasmodium falciparum. Science 263, 961--963.
           		
           		\bibitem{Ferreira2004} Ferreira, M.U., da Silva Nunes, M., Wunderlich, G., 2004. Antigenic diversity and immune evasion by malaria parasites. Clin. Diagn. Lab. Immunol. 11, 987--995. 
           		
           		\bibitem{Gupta1996}Gupta, S., Maiden, M.C.J., Feavers, I.M., Nee, S., May, R.M., Anderson, R.M., 1996. The maintenance of strain structure in populations of recombining infectious agents. Nat. Med. 2, 437--442. 
           		
           		\bibitem{Gupta1999} Gupta, S., Anderson, R.M., 1999. Population structure of pathogens: the role of immune selection. Parasitol. Today 15, 497--501.
           		
           		\bibitem{Gog2002} Gog, J.R., Grenfell, B.T., 2002. Dynamics and selection of many-strain pathogens. Proc. Natl. Acad. Sci. USA 99, 17209--17214.
           		
           		\bibitem{Ferguson2003} Ferguson, N.M., Galvani, A.P., Bush, R.M., 2003. Ecological and immunological determinants of influenza evolution. Nature 422, 428--33. 
           		
           		\bibitem{Levin2004} Levin, S.A., Dushoff, J., Plotkin, J.B., 2004. Evolution and persistence of influenza A and other diseases. Math. Biosci. 188, 17--28. 
           		
           		\bibitem{Recker2009} Recker, M., Blyuss, K.B., Simmons, C.P., Hien, T.T., Wills, B., Farrar, J., Gupta, S., 2009. Immunological serotype interactions and their effect on the epidemiological pattern of dengue. Proc. Roy. Soc. B 276, 2541--2548. 
           		
           		\bibitem{Buckee2004} Buckee, C.O., Koelle, K., Mustard, M.J., Gupta, S., 2004. The effects of host contact network structure on pathogen diversity and strain structure. Proc. Natl. Acad. Sci. USA 101, 10839--10844.
           		
           		\bibitem{Buckee2010} Buckee, C.O., Gupta, S., 2010. A network approach to understanding pathogen population structure, in Sintchenko, V. (Ed.), Infectious Disease Informatics. Springer New York, 167--185. 
           		
           		\bibitem{Cisternas2004} Cisternas, J., Gear, C.W., Levin, S., Kevrekidis, I.G., 2004. Equation-free modelling of evolving diseases: coarse-grained computations with individual-based models. Proc. R. Soc. A 460, 2761--2779. 
           		
           		\bibitem{Andreasen1997} Andreasen, V., Lin, J., Levin, S.A., 1997. The dynamics of cocirculating influenza strains conferring partial cross-immunity. J. Math. Biol. 35, 825--842.
           		
           		\bibitem{Gomes2002} Gomes, M.G.M., Medley, G.F., Nokes, D.J., 2002. On the determinants of population structure in antigenically diverse pathogens. Proc. Roy. Soc. B 269, 227--233. 
           		
           		\bibitem{Kryazhimskiy2007} Kryazhimskiy, S., Dieckmann, U., Levin, S.A., Dushoff, J., 2007. On state-space reduction in multi-strain pathogen models, with an application to antigenic drift in influenza A. PLoS Comput. Biol. 3, e159. 
           		
           		\bibitem{Zhou2012} Zhou, C., Zhou, Y., 2012. Strategies for viral cross protection in plants. Meth. Mol. Biol. 894, 69--81. 
           		
           		\bibitem{Pennazio2001} Pennazio, S., Roggero, P., Conti, M., 2001. A history of plant virology. Cross protection. New Microbiol.
		24, 99--114.
           		
           		\bibitem{Beachy1999} Beachy, R.N., 1999. Coat-protein-mediated resistance to tobacco mosaic virus: discovery mechanisms and exploitation. Phil. Trans. Roy. Soc. Lond. B 354, 659--664. 
           		
           		\bibitem{Bendahmane1997} Bendahmane, M., Fitchen, J.H., Zhang, G., Beachy, R.N., 1997. Studies of coat protein-mediated resistance to tobacco mosaic tobamovirus: correlation between assembly of mutant coat proteins and resistance. J. Virol. 71, 7942--7950.
           		
           		\bibitem{Lee2005} Lee, Y., Tscherne, D., Yun, S., 2005. Dual mechanisms of pestiviral superinfection exclusion at entry and RNA replication. J. Virol. 79, 3231--3242.
           		
           		\bibitem{Takeshita2004} Takeshita, M., Shigemune, N., Kikuhara, K., Furuya, N., Takanami, Y., 2004. Spatial analysis for exclusive interactions between subgroups I and II of Cucumber mosaic virus in cowpea. Virology 328, 45--51. 
           		
           		\bibitem{Gal-On2006} Gal-On, A., Shiboleth, Y., 2006. Cross-protection, in: Loebenstein, G., Carr, J.P. (Eds), 
Natural Resistance Mechanisms of Plants to Viruses, Kluwer, Dordrecht, 261--288
           		
           		\bibitem{Ratcliff1999} Ratcliff, F.G., MacFarlane, S.A., Baulcombe, D.C., 1999. Gene silencing without DNA: RNA-mediated cross-protection between viruses. The Plant Cell 11, 1207--1215.
           		
           		\bibitem{Waterhouse1999} Waterhouse, P.M., Smith, N.A., Wang, M.B., 1999. Virus resistance and gene silencing: killing the messenger. Trends Plant Sci. 4, 452--457.
           	
           		\bibitem{Escobar2000} Escobar, M.A., Dandekar, A.M., 2003. Post-transcriptional gene silencing in plants, in Barciszewski, J., Erdmann, V.A. (Eds), Noncoding RNAs: molecular biology and molecular medicine, Kluwer Academic, Dordrecht, 129--140.
           		
           		\bibitem{Sijen2000} Sijen, T., Kooter, J.M., 2000. Post-transcriptional gene-silencing: RNAs on the attack or on the defense? BioEssays 22, 520--531.
           		
           		\bibitem{AlessandraTenorioCosta2013} Costa, A.T., Bravo, J.P., Makiyama, R.K., Vasconcellos Nunes, A., Maia, I.G., 2013. Viral counter defines X antiviral immunity in plants: mechanisms for survival, in Romanowski, V., Current issues in molecular virology - viral genetics and biotechnological applications, InTech.
           		
           		\bibitem{Hammond2000} Hammond, S.M., Bernstein, E., Beach, D., Hannon, G.J., 2000. An RNA-directed nuclease mediates post-transcriptional gene silencing in Drosophila cells. Nature 404, 293--296.
           		
           		\bibitem{Bernstein2001} Bernstein, E., Caudy, A.A., Hammond, S.M., Hannon, G.J., 2001. Role for a bidentate ribonuclease in the initiation step of RNA interference. Nature 409, 363--366.
           		
           		\bibitem{Wassenegger2000} Wassenegger, M., 2000. RNA-directed DNA methylation. Plant Mol. Biol. 43, 203--220.
           		
           		\bibitem{Zhang2012b} Zhang, C., Ruvkun, G., 2012. New insights into siRNA amplification and RNAi. RNA Biol. 9, 1045--1049.
           		%
           		\bibitem{Pumplin2013} Pumplin, N., Voinnet, O., 2013. RNA silencing suppression by plant pathogens: defence, counter-defence and counter-counter-defence. Nature Rev. Microbiol. 11, 745--760.
           		
           		\bibitem{Raja2008} Raja, P., Sanville, B.C., Buchmann, R.C., Bisaro, D.M., 2008. Viral genome methylation as an epigenetic defense against geminiviruses. J. Virol. 82, 8997--9007.
           		
           		\bibitem{Malapi-Nelson2009} Malapi-Nelson, M., Wen, R.-H., Ownley, B.H., Hajimorad, M.R., 2009. Co-infection of soybean with Soybean mosaic virus and Alfalfa mosaic virus results in disease synergism and alteration in accumulation level of both viruses. Plant Disease 93, 1259--1264. 		
           		
           		\bibitem{Wege2007} Wege, C., Siegmund, D., 2007. Synergism of a DNA and an RNA virus: enhanced tissue infiltration of the begomovirus Abutilon mosaic virus (AbMV) mediated by Cucumber mosaic virus (CMV). Virology 357, 10--28.
           		%
           		\bibitem{Pruss1997} Pruss, G., Ge, X., Shi, X. M., Carrington, J. C.,  Vance, V. B., 1997. Plant viral synergism: the potyviral genome encodes a broad-range pathogenicity enhancer that transactivates replication of heterologous viruses. Plant Cell 9, 859--868.	
           		
           		\bibitem{Roossinck2005} Roossinck, M.J., 2005. Symbiosis versus competition in plant virus evolution. Nature Rev. Microbiol. 3, 917--924.	
           		
           		\bibitem{Takeshita2005} Takeshita, M., 2005. Molecular biological study of host specificity and cross-protection of Cucumber mosaic virus. J. Gen. Plant Path. 71, 459.
           		
           		\bibitem{Bergua2014} Bergua, M., Zwart, M.P., El-Mohtar, C., Shilts, T., Elena, S.F., Folimonova, S.Y., 2014. A viral protein mediates superinfection exclusion at the whole-organism level but is not required for exclusion at the cellular level. J. Virol. 88, 11327--11338.
           		
           		\bibitem{Reddy2012} Chowda Reddy, R.V., Dong, W., Njock, T., Rey, M.E.C., Fondong, V.N., 2012. Molecular interaction between two cassava geminiviruses exhibiting cross-protection. Virus res. 163, 169--177.
           		
           		\bibitem{Zhang2001} Zhang, X.-S., Holt, J., 2001. Mathematical models of cross protection in the epidemiology of plant-virus diseases. Phytopath. 91, 924--934.
           		
           		\bibitem{Jeger2011} Jeger, M.J., van den Bosch, F., Madden, L.V., 2011. Modelling virus- and host-limitation in vectored plant disease epidemics. Virus res. 159, 215--222.
           		
           		\bibitem{per02} Perelson, A.S., 2002. Modelling viral and immune system dynamics. Nature Rev. Imm. 2, 28--36.
			
			\bibitem{wod02} Wodarz, D., Christensen, J.P., Thomsen, A.R., 2002. Trends Imm. 23, 194--200.
		
           		\bibitem{Mauck2010}Mauck, K.E., De Moraes, C.M., Mescher, M.C., 2010. Deceptive chemical signals induced by a plant virus attract insect vectors to inferior hosts. Proc. Natl. Acad. Sci. USA 107, 3600--3605. 
           		
           		\bibitem{Moreno-Delafuente2013} Moreno-Delafuente, A., Garzo, E., Moreno, A., Fereres, A., 2013. A plant virus manipulates the behavior of its whitefly vector to enhance its transmission efficiency and spread. PLoS One 8, e61543. 
           		
           		\bibitem{Lalonde2004} Lalonde, S., Wipf, D., Frommer, W.B., 2004. Transport mechanisms for organic forms of carbon and nitrogen between source and sink. Ann. Rev. Plant Biol. 55, 341--372.
           		
           		\bibitem{Opalka1998} Opalka, N., Brugidou, C., Bonneau, C., Nicole, M., Beachy, R.N., Yeager, M., Fauquet, C., 1998. Movement of rice yellow mottle virus between xylem cells through pit membranes. Proc. Natl. Acad. Sci. USA 95, 3323--3328.
           		
           		\bibitem{Wan2015} Wan, J., Cabanillas, D.G., Zheng, H., Lalibert\'e, J.-F., 2015. Turnip mosaic virus moves systemically through both phloem and xylem as membrane-associated complexes. Plant Physiol. 167, 1374--1388. 
           		
           		\bibitem{Paine2012}Paine, C. E. T., Marthews, T. R., Vogt, D. R., Purves, D., Rees, M., Hector, A., Turnbull, L. A. (2012). How to fit nonlinear plant growth models and calculate growth rates: An update for ecologists. Meth. Ecol. Evol. 3, 245--256. 
           		
		\bibitem{Heinen1999}Heinen, M. (1999). Analytical growth equations and their Genstat 5 equivalents. Netherlands J. Agric. Sci. 47,
		67--89.

			\bibitem{tridane} Tridane, A., Kuang, Y. (2010). Modeling the interaction of cytotoxic T lymphocytes and influenza virus infected epithelial cells. Math. Biosci. Eng. 7, 175--189.
			
			\bibitem{PerNel} Perelson, A.S., Nelson. P.W. (1999). Mathematical analysis of HIV-1 dynamics in vivo. SIAM Rev. 41, 3--44.
			
			\bibitem{ciupe07} Ciupe, S.M., Ribeiro, R.M., Nelson, P.W. Perelson, A.S. (2007). Modeling the mechanisms of acute hepatitis B virus infection. J. Theor. Biol. 247, 23--35.
	           		
           		\bibitem{Zvereva2012} Zvereva, A.S., Pooggin, M.M., 2012. Silencing and innate immunity in plant defense against viral and non-viral pathogens. Viruses 4, 2578--2597. 
           		
           		\bibitem{Hinrichs1998} Hinrichs, J., Berger, S., Shaw, J.G., 1998. A hypersensitive response-like mechanism is involved in resistance of potato plants bearing the Ry(sto) gene to the potyviruses potato virus Y and tobacco etch virus. J. Gen. Virol. 79, 167--176.
           		
           		\bibitem{Fritig2007} Fritig, B., Kauffmann, S., 2007. Mechanism of the hypersensitivity reaction of plants. Ciba Found. Symp. 133 - Plant Resist. to Virus.
           		
           		\bibitem{Burdon2003} Burdon, J.J., Thrall, P.H., 2003. The fitness costs to plants of resistance to pathogens. Genome Biol. 4, 227. 
           		
           		\bibitem{Tian2003} Tian, D., Traw, M.B., Chen, J.Q., Kreitman, M., Bergelson, J., 2003. Fitness costs of R-gene-mediated resistance in Arabidopsis thaliana. Nature 423, 74--77.
           		
		\bibitem{Hethcote2000} Hethcote, H.W., 2000. The mathematics of infectious diseases. SIAM Rev. 42, 599--653. 
		
           		\bibitem{Dietz1993} Dietz, K., 1993. The estimation of the basic reproduction number for infectious diseases. Stat. Methods Med. Res. 2, 23--41.
           		
           		\bibitem{Driessche2008} van den Driessche, P., Watmough, J., 2008. Further notes on the basic reproduction number, in Brauer, F., van den Driessche, P., Wu, J. (eds), Mathematical Epidemiology, Springer,  Berlin.
           		
           		\bibitem{Heesterbeek1996} Heesterbeek, J.A.P., Dietz, K., 1996. The concept of $R_0$ in epidemic theory. Stat. Neerl. 50, 89--110. 
           		           		
           		\bibitem{Allen2003} Allen, L.J.S., Langlais, M., Phillips, C.J., 2003. The dynamics of two viral infections in a single host population with applications to hantavirus. Math. Biosci. 186, 191--217.
           		
           		\bibitem{Castillo-Chavez1996} Castillo-Chavez, C., Huang, W., Li, J., 1996. Competitive exclusion in gonorrhea models and other sexually transmitted diseases. SIAM J. Appl. Math. 56, 494--508. 
           		
           		\bibitem{Melnyk2011} Melnyk, C.W., Molnar, A., Baulcombe, D.C., 2011. Intercellular and systemic movement of RNA silencing signals. EMBO J. 30, 3553--3563. 
           		
           		\bibitem{Liang2012} Liang, D., White, R.G., Waterhouse, P.M., 2012. Gene silencing in Arabidopsis spreads from the root to the shoot, through a gating barrier, by template-dependent, nonvascular, cell-to-cell movement. Plant Physiol. 159, 984--1000. 
           		
           		\bibitem{Himber2015} Himber, C., Dunoyer, P., 2015. The tracking of intercellular small RNA movement. Methods Mol. Biol. 1217, 275--281. 
           		
           		\bibitem{Kawakami2004} Kawakami, S., Watanabe, Y., Beachy, R.N., 2004. Tobacco mosaic virus infection spreads cell to cell as intact replication complexes. Proc. Natl. Acad. Sci. USA 101, 6291--6296. 
           		
           		\bibitem{Silva2002} Silva, J.M., Hammond, S.M., Hannon, G.J., 2002. RNA interference: a promising approach to antiviral therapy? Trends Mol. Med. 8, 505--508. 
           		
           		\bibitem{Caplen2004} Caplen, N.J., 2004. Gene therapy progress and prospects. Downregulating gene expression: the impact of RNA interference. Gene Ther. 11, 1241--1248. 
           		
           		\bibitem{Soifer2007}Soifer, H.S., Rossi, J.J., Saetrom, P., 2007. MicroRNAs in disease and potential therapeutic applications. Mol. Ther. 15, 2070--2079. 
           		
           		\bibitem{Neofytou2016}Neofytou, G., Kyrychko, Y.N., Blyuss, K.B., 2016. Time-delayed model of immune response in plants. J. Theor. Biol. 389, 28--39.
           		
           	\end{thebibliography}
 		\end{document}